\documentclass{llncs}
\usepackage[utf8]{inputenc}
\usepackage{amssymb, amsfonts, amsmath, bm, mathtools}
\usepackage{url}
\usepackage{fullpage}
\usepackage{xspace}
\usepackage{tikz,pgfplots,xcolor}
\usepackage{color}
\usepackage{hyperref}
\usepackage{array}
\usepackage{enumitem}
\allowdisplaybreaks

\usepackage{amsthm} 	
\usepackage{thmtools} 	

\definecolor{light-gray}{gray}{0.95}

\declaretheoremstyle[
    spaceabove=6pt, spacebelow=6pt,
    headfont=\normalfont\bfseries,
    bodyfont=\normalfont\itshape,
    notefont=\normalfont\bfseries, notebraces={(}{)},
    postheadspace=.5em,
    mdframed={
  		hidealllines=true,
  		skipabove=10pt,
  		backgroundcolor=light-gray,
		innerrightmargin=2pt,
		innerleftmargin=2pt,
		innertopmargin=0pt,
		innerbottommargin=2pt,
 		leftmargin=0pt,
 		rightmargin=0pt,
 	 }
]{mathstyle}

\declaretheorem[name={Theorem}, 
    style=mathstyle
]{thm}

\declaretheorem[name={Lemma}, 
    style=mathstyle
]{lem}

\declaretheorem[name={Definition},
    style=mathstyle
]{defi}

\title{Access Control Encryption: \\ Enforcing Information Flow with Cryptography
\thanks{This project was supported by: the Danish National Research Foundation and The National Science Foundation of China (grant 61361136003) for the Sino-Danish Center for the Theory of Interactive Computation.}
\thanks{Full version of TCC 2016-B paper~\cite{ACE-TCC2016B}}}
\author{Ivan Damg{\aa}rd \and Helene Haagh \and Claudio Orlandi}
\institute{\{ivan,haagh,orlandi\}@cs.au.dk, Aarhus University}
\begin{document}

\maketitle
\newcommand{\todo}[1]{\textcolor{red}{(TODO: #1)}}
\newcommand{\hfh}[1]{\textcolor{green}{(Helene: #1)}}
\newcommand{\cnote}[1]{\textcolor{blue}{(Claudio: #1)}}

\newcommand{\from}{\leftarrow}
\newcommand{\zo}{\{0,1\}}
\newcommand{\zon}{\{0,1\}^n}
\newcommand{\zok}{\{0,1\}^\kappa}

\newcommand{\RR}{\mathbb{R}}
\newcommand{\NN}{\mathbb{N}}
\newcommand{\San}{\mathsf{San}}
\newcommand{\negl}{\mathsf{negl}}

\newtheorem{construction}{Construction}

\newcommand{\Setup}{\mathsf{Setup}}
\newcommand{\Gen}{\mathsf{Gen}}
\newcommand{\Enc}{\mathsf{Enc}}
\newcommand{\Dec}{\mathsf{Dec}}
\newcommand{\randDec}{\mathsf{RDec}}
\newcommand{\Prove}{\mathsf{Prove}}
\newcommand{\Verify}{\mathsf{Verify}}
\newcommand{\M}{\mathcal{M}}
\newcommand{\sen}{\mathsf{sen}}
\newcommand{\rec}{\mathsf{rec}}
\newcommand{\san}{\mathsf{san}}
\newcommand{\C}{\mathcal{C}}
\newcommand{\Sim}{\mathsf{Sim}}

\newcommand{\Rand}{\San}
\newcommand{\PKE}{\mathsf{PKE}}
\newcommand{\secparam}{\kappa}
\newcommand{\ACE}{\mathsf{ACE}}

\newcommand{\desc}{\textup{desc}}

\newcommand{\oACE}{$1$-ACE\xspace}

\newcommand{\oneACE}{\mathsf{1ACE}}
\newcommand{\oneACEpow}{\mathsf{1ACE}}
\newcommand{\EGACE}{\mathsf{EGACE}}

\newcommand{\CCG}{\mathsf{Challenge}}

\newcommand{\sanFE}{\mathsf{sFE}}
\newcommand{\sanFEpow}{\mathsf{sFE}}
\newcommand{\sanPKE}{\mathsf{sPKE}}
\newcommand{\NIZK}{\mathsf{NIZK}}
\newcommand{\PRF}{\mathsf{PRF}}

\newcommand{\R}{\mathcal{R}}
\newcommand{\MasterDec}{\mathsf{MDec}}
\newcommand{\adv}[1]{\mathsf{adv}^{#1}} 

\newcommand{\sssprob}{p_{sss}}

\newcounter{claimcounter}
\renewenvironment{claim}{
	\refstepcounter{claimcounter}
	\medskip
	\noindent \textit{Claim~\arabic{claimcounter}}. \enskip
}{

}

\begin{abstract}
We initiate the study of \emph{Access Control Encryption (ACE)}, a novel cryptographic primitive that allows fine-grained access control, by giving different rights to different users not only in terms of which messages they are allowed to \emph{receive}, but also which messages they are allowed to \emph{send}.

Classical examples of security policies for information flow are the well known Bell-Lapadula~\cite{lapadula1996secure} or Biba~\cite{biba} model: in a nutshell, the Bell-Lapadula model assigns roles to every user in the system  (e.g., \emph{public}, \emph{secret} and \emph{top-secret}). A users' role specifies which messages the user is allowed to receive (i.e., the \emph{no read-up} rule, meaning that users with \emph{public} clearance should not be able to read messages marked as \emph{secret} or \emph{top-secret}) but also which messages the user is allowed to send (i.e., the \emph{no write-down} 
rule, meaning that a malicious user with \emph{top-secret} clearance should not be able to write messages marked as 
\emph{secret} or \emph{public}). To the best of our knowledge, no existing cryptographic primitive allows for even this 
simple form of access control, since no existing cryptographic primitive enforces any restriction on what 
kind of messages one should be able to encrypt. Our contributions are:
\begin{itemize}
\item Introducing and formally defining access control encryption (ACE);
\item A construction of ACE with complexity linear in the number of the roles based on classic number theoretic assumptions (DDH, Paillier);
\item A construction of ACE with complexity polylogarithmic in the number of roles based on recent results 
on cryptographic obfuscation;
\end{itemize}
\end{abstract}

\section{Introduction}

Traditionally, cryptography has been about providing secure communication over insecure channels. We want to 
protect honest parties from external adversaries: only the party who has the decryption key can access 
the message. More recently, more complicated situations have been considered, where we do not want to trust 
everybody with the same information: depending on who you are and which keys you have, you can access 
different parts of the information sent (this can be done using, e.g., functional 
encryption~\cite{DBLP:conf/tcc/BonehSW11}).

However, practitioners who build secure systems in real life are often interested in achieving different and 
stronger properties: one wants to control the information flow in the system, and this is not just about 
what you can receive, but also about what you can send. As an example, one may think of the first security 
policy model ever proposed, the one by Bell and Lapadula \cite{lapadula1996secure}. Slightly simplified,  
this model classifies users of a system in a number of levels, from ``public''  in the bottom to 
``top-secret'' on top. Then two rules are defined: 1) ``no read-up'' -- a user is not allowed to receive
 data from higher levels and 2) ``no write-down'' -- a user is not allowed to send data to lower levels. 
 The idea is of course to ensure confidentiality: data can flow from the bottom towards the top, but not in
 the other direction. Clearly, 
both rules are necessary, in particular we need no write-down, since a party on top-secret level may try to 
send information she should not, either by mistake or because her machine has been infected by a virus.

In this paper we study the question of whether cryptography can help in enforcing such security policies.
A first thing to realize is that this problem cannot be solved without some assumptions about physical, 
i.e.,  non-cryptographic security: if the communication lines cannot be controlled, we cannot prevent a 
malicious user 
from sending information to the wrong place. We therefore must introduce a party that controls
the communication, which we will call the 
\emph{sanitizer} $\San$. We assume that all outgoing communication must pass through this party. 
$\San$ can then be instructed to do some specific processing on the messages it gets.

Of course, with this assumption the problem can be solved: $\San$ is told what the security policy is and simply blocks all messages that should not be sent according to the policy. This is actually a (simplified) model of how existing systems work, where $\San$ is implemented by the operating system and various physical security measures.

However, such a solution is problematic for several reasons: users must securely identify themselves to
$\San$ so that he can take the correct decisions, this also means that when new users join the system
$\San$ must be informed about this, directly or indirectly.  A side effect  of this is that $\San$ necessarily 
knows who sends to whom, and must of course know the security policy. This means that a company cannot outsource the function of $\San$ to  another party without disclosing information on internal activities of the company.

Therefore, a better version of our basic question is the following: \emph{can we use cryptography to simplify
the job of $\San$ as much as is possible, and also ensure that he learns minimal information?} 

To make the goal more precise, note that it is clear that $\San$ must process each message that is sent
i.e., we cannot allow a message violating the policy to pass through unchanged. But we can hope that the processing to be done does not depend on the security policy, and also not on the identities of the sender and therefore of the allowed receivers. This way we get rid of the need for users to identify themselves to $\San$. It is also clear that $\San$ must at least learn when a message was sent and its length, but we can hope to ensure he learns nothing more. This way, one can outsource the function of running $\San$ to a party that is only trusted to execute correctly.

Our goal in this paper is therefore to come up with a cryptographic notion and a construction that  reduces the sanitizer's job to the minimum we just described.
To the best of our knowledge, this problem has not been studied before in the cryptographic literature, and it is easy to see that existing constructions only solve ``half the problem'': we can easily control which users you can \emph{receive from} by selecting the key material we give out (assuming that the sender is honest). This is exactly what attribute based \cite{DBLP:conf/ccs/GoyalPSW06} or functional encryption \cite{DBLP:conf/tcc/BonehSW11} can do. But any such scheme of course allows a malicious sender to encrypt what he wants for any receiver he wants.

\paragraph{Our Contribution.} In this paper we propose a solution based on
a new notion called Access Control Encryption (ACE).
In a nutshell ACE works as follows:  an ACE scheme has a key generation algorithm that produces a set of sender keys, a set of receiver keys and a sanitizer key. 
An honest  sender $S$  encrypts  message $m$ under his sender key and sends it for processing by $\San$ using the sanitizer key. $\San$ does not need to know the security policy, nor who sends a message or where it is going (so a sender does not have to identify himself): $\San$ simply executes a specific randomised algorithm on the incoming ciphertext and passes the result on to a broadcast medium, e.g., a disk from where all receivers can read. So, as desired, 
$\San$ only knows when a message was sent and its length.\footnote{Note that the sanitizer has to send the ciphertext to all receivers -- both those who are allowed to decrypt and those who are not. A sanitizer who could decide whether a particular receiver is allowed to receive a particular ciphertext would trivially be able to distinguish between different senders with different writing rights.} An honest receiver $R$ who is allowed to receive from $S$ is able to recover $m$ 
using his key and the output from $\San$. On the other hand, consider a corrupt sender $S$ who is not allowed to send to $R$. ACE ensures that no matter what $S$ sends, what $R$ receives (after being processed by $\San$) looks like a random encryption of a random message.
 In fact we achieve security against collusions: considering a subset $\cal S$ of senders and a subset $\cal R$ of receivers, if none of these senders are allowed to send to any of the receivers, then $\cal S$ cannot transfer any information to 
$\cal R$, even if players in each set work together. We propose two constructions of ACE: one based on standard number theoretic assumptions (DDH, Pailler) which achieves complexity linear in the number of roles,  and one based on recent results in cryptographic obfuscation, which achieves complexity polylogarithmic in the number of roles.

\paragraph{Example.} A company is working on a top-secret military project for the government. To protect the secrets the company sets up an access policy that determines which employees are allowed to communicate (e.g., a researcher with top-secret clearance should not be allowed to send classified information to the intern, who is making the coffee and only has public clearance). 
To implement the access policy, the company sets up a special server that sanitizes every message sent on the internal network before publishing it on a bulletin board or broadcasting it. Using ACE this can be done without requiring users to log into the sanitizer. Furthermore,  if corrupted parities (either inside or outside the company) want to intercept the communication they will get no information from the sanitizer server, since it does not know the senders identities and the messages sent over the network.

In the following sections, we describe ACE in more detail and take a closer look at our technical contributions.

\subsection{Access Control Encryption: The Problem it Solves} 

\paragraph{Senders and Receivers.} We have $n$ (types of) senders $S_1,\ldots,S_n$ and $n$ (types of) receivers $R_1,\ldots,R_n$.\footnote{The number of senders equals the number of receivers only for the sake of exposition.} There is some predicate $P : [n] \times [n] \to \zo$, where $P(i,j)=1$ means that $S_i$ is allowed to send to $R_j$, while $P(i,j)=0$ means that $S_i$ is not allowed to send to $R_j$. 

\paragraph{Network Model.} We assume that senders are connected to all receivers via a public channel i.e., a sender cannot send a message only to a specific receiver and any receiver can see all traffic from all senders (also from those senders they are not allowed to communicate with). 

\paragraph{Requirements.} Informally we want the following properties\footnote{The security model, formalized in Definitions~\ref{def:ACEnoread} and~\ref{def:ACEnowrite}, is more general than this.} 
\begin{enumerate}
\item \emph{Correctness:} When an honest sender $S_i$ sends a message $m$, all receivers $R_j$ with $P(i,j)=1$ learn $m$;
\item \emph{No-Read Rule:} At the same time all receivers $R_j$ with $P(i,j)=0$ learn no information about $m$;
\item \emph{No-Write Rule:} No (corrupt) sender $S_i$ should be able to communicate any information to any (possibly corrupt) receiver $R_j$ if $P(i,j)=0$
\end{enumerate}

Note that the \emph{no-read rule} on its own is a simple \emph{confidentiality} requirement, which can be enforced using standard encryption schemes. On the other hand standard cryptographic tools do not seem to help in satisfying the \emph{no-write rule}.
In particular the \emph{no-write} rule is very different from the standard \emph{authenticity} requirement and e.g., signature schemes cannot help here: had we asked for a different property such as \emph{``a corrupt sender $S_i$ should not be allowed to communicate with an honest receiver $R_j$ if $P(i,j)=0$''} then the problem could be solved by having $R_j$ verify the identity of the sender (using a signature scheme) and ignore messages from any sender $i$ with $P(i,j)=0$. Instead, we are trying to block communication even between corrupt senders and corrupt receivers. 

The problem as currently stated is impossible to solve, since a corrupt sender can broadcast $m$ in the clear to all receivers (the corrupt sender might not care that other receivers also see the message). As mentioned above, we therefore enhance the model by adding a special party, which we call the \emph{sanitizer} $\San$. The sanitizer receives messages from senders, performs some computation on them, and then forwards them to all receivers. 
In other words, we allow the public channel to perform some computation before delivering the messages to the receivers. Hence, the output of the sanitizer is visible to all receivers (i.e., the sanitizer cannot give different outputs to different receivers). 
We therefore add the following requirement to our \emph{no-read rule}:
\begin{enumerate}
\item[2b.] The sanitizer should not learn anything about the communication it routes. In particular, the sanitizer should not learn any information about the message $m$ which is being transmitted nor the identity of the sender $i$; 
\end{enumerate}

In Section~\ref{sec:definition} we formalize properties $2$ and $2b$ as a single one (i.e., no set of corrupt receivers, even colluding with the sanitizer, should be able to break the \emph{no-read} rule). When considering property $3$, we assume the sanitizer not to collude with the corrupt senders and receivers: after all, since the sanitizer controls the communication channel, there is no way of preventing a corrupt sanitizer from forwarding messages from corrupt senders to the corrupt receivers.\footnote{Note that it is possible to reduce the trust on the sanitizer in different ways: in a black-box way, one could imagine several parties emulating the work of the sanitizer using MPC. In a more concrete way, it is possible to have a \emph{chain} of sanitizers, where the senders send their encryptions to sanitizer $1$, the receivers receive ciphertexts from sanitizer $n$, and sanitizer $i+1$ further sanitizes the output of sanitizer $i$. We note that all definitions and constructions in this paper can be easily generalized to this scenario but, to keep the presentation as simple as possible, we do not discuss this solution further and stick to the case of a single sanitizer.}

We stress that previous work is not sufficient to achieve property $3$: Even encryption schemes with fine-grained decryption capabilities (such as predicate- and attribute based- encryption \cite{DBLP:conf/ccs/GoyalPSW06,DBLP:journals/joc/KatzSW13}) do not offer security guarantees against colluding senders and receivers.


\subsection{Technical Overview} 

\paragraph{Linear ACE.} The main idea behind our construction of ACE with linear complexity (described in Section~\ref{sec:linear}) 
is the following: we start with an ACE for a single identity i.e., where $n=1$ and $P(1,1)=1$. First we need 
to make sure that even a corrupt sender with encryption rights (i.e., $i=1$) cannot communicate with a 
corrupt receiver with no decrypting right (i.e., with a special identity $j=0$). To prevent this, since the receiver cannot 
decrypt the ciphertext, it is enough to use a randomizable public key encryption and let the sanitizer 
refresh the ciphertext. This ensures that the outgoing ciphertext is distributed exactly as a fresh 
encryption.

The more challenging task is to ensure that a corrupt sender with no rights (i.e., with a special identity $i=0$) cannot transfer any information to a corrupt receiver with decrypting rights (i.e., $j=1$), since in this case the receiver knows the decryption key. Thus, we cannot use the security of the underlying encryption scheme. We solve the problem using any encryption scheme which is homomorphic both in the message and in the randomness (such as ElGamal or Pailler). The main idea is to let the encryption key $ek$ as well as the randomizer key $rk$ be some secret value $\alpha$, and an encryption of a message $m$ being a tuple $(c_0,c_1)=(E(ek),E(m))$. On input such a tuple the sanitizer picks a random $s$ and outputs $c'$, a fresh encryption of $(ek-rk)\cdot s +m$ (which can be computed thanks to the homomorphic properties of $E$): note that sanitization does not interfere with honestly generated encryptions (since $ek=rk=\alpha$), while the sanitized version of a ciphertext produced by anyone who does not know $\alpha$ is indistinguishable from a random encryption of a random value. 

We then turn this into a scheme for any predicate $P : [n]\times[n] \to \zo$ by generating $n$ copies of the single identity ACE scheme. Each receiver $j$ is given the decryption key for one of the schemes, and each sender $i$ is given the encryption key for all instances $j$ such that $P(i,j)=1$. The resulting scheme has linear complexity in $n$, the number of the roles in the system, which makes our scheme impractical for large predicates.


\paragraph{Polylogarithmic ACE.} At first it might seem easy to construct an ACE scheme with compact ciphertexts using standard tools (such as non-interactive zero-knowledge proofs). In Section~\ref{sec:polylog} we discuss why this is not the case before presenting our construction of an ACE with complexity polylogarithmic in $n$. To construct the scheme we first introduce the notion of a \emph{sanitizable functional encryption} (sFE) scheme which is a functional encryption (FE) scheme enhanced with a sanitization algorithm. Informally we require that given any two ciphertexts $c_0,c_1$ that decrypt to the same message and a sanitized ciphertext $c'$, no one (even with access to the master secret key), should be able to tell whether $c'$ is a sanitized version of $c_0$ or $c_1$.\footnote{We note that this is a relaxation of re-randomizability for FE, in the sense that we do not require sanitized ciphertexts to be indistinguishable from fresh encryptions, but only independent of the randomness used in the original encryption. However, to the best of our knowledge, no re-randomizable FE scheme for all circuits exist.} We are able to construct such a scheme by modifying the FE based on indistinguishability obfuscation of Garg et al.~\cite{DBLP:conf/focs/GargGH0SW13}: in their scheme ciphertexts consist of two encryptions and a \emph{simulation statistically-sound} NIZK proof that they contain the same message. We instantiate their construction with a sanitizable encryption scheme\footnote{Similar to a re-randomizable encryption scheme, where we do not require sanitized ciphertexts to look indistinguishable from fresh encryptions.}, and we instruct the sanitizer to sanitize the two encryptions, drop the original proof and append a \emph{proof of a proof} instead (that is, a proof of the fact that the sanitizer saw a proof who would make the original verifier accept). This preserves the functionality of the original FE scheme while making the sanitized ciphertext independent of the randomness used by the sender. 
We formally define sFE in Section~\ref{sec:sfe} and present a construction in Section~\ref{sec:sfecons}.

Finally, armed with such a sFE scheme, we construct a polylog ACE scheme in Section~\ref{sec:polylogace} in the following way: ciphertexts are generated by encrypting tuples of the form $(m,i,y)$ with $y=F_{ek_i}(m)$ for a PRF $F$ (where $ek_i$ is the the encryption key of the sender $S_i$), using the sFE scheme. Decryption keys are sFE secret keys for the function that outputs $m$ only if $P(i,j)=1$ (and ignores $y$). The sanitizer key is a sFE secret key which outputs $1$ only if $y$ is a valid MAC on $m$ for the identity $i$ (note that this can be checked by a compact circuit by e.g., generating all the keys $ek_i$ pseudorandomly using another PRF). This key allows the sanitizer to check if an encryption contains a valid MAC or not, but without learning anything about the message nor the identity. Now the sanitizer drops invalid encryptions (or replaces them with random encryptions of random values for a special, undecryptable identity $i=0$) and forwards valid encryptions (after having refreshed them).

\paragraph{Open Questions.} We identify two major opens questions: the first one is to construct practically interesting ACE from noisy, post-quantum assumptions such as LWE -- the challenge here is that it always seems possible for a malicious sender to encrypt with just enough noise that any further manipulation by the sanitizer makes the decryption fail. This can be addressed using ``bootstrapping'' techniques, but this is not likely to lead to schemes with efficiency comparable to the ones based on DDH or Pailler described above. The second open question is to design sublinear ACE scheme with practical efficiency even for limited classes of interesting predicates such as e.g., $P(i,j)=1 \Leftrightarrow i \geq j$.

\subsection{Related Work} 
One of the main challenges in our setting is to prevent corrupt senders to communicate to corrupt receivers using subliminal channels (e.g., by producing the encryptions with maliciously generated randomness). In some sense we are trying to prevent steganography~\cite{DBLP:conf/crypto/HopperLA02}. 
Recent work on cryptographic firewalls~\cite{DBLP:conf/eurocrypt/MironovS15,DBLP:journals/iacr/DodisMS15} also deals with this problem, but in the context of preventing malicious software implementations to leak information via steganographic techniques. 
Raykova et al.~\cite{DBLP:conf/fc/RaykovaZB12} presented solutions to the problem of access control on outsourced data, with focus on hiding the access patterns from the cloud (this is not a concern in our application since all receivers receive all ciphertexts) and in preventing malicious writers from updating files they are not allowed to update. However they only guarantee that malicious writers are caught if they do so, while we want to prevent any communication between corrupt senders and receivers. 
Backes and Pfitzmann introduced the notion of probabilistic non-interference which allows to relate cryptography to the notion of information flow for both transitive~\cite{DBLP:conf/sp/BackesP03} and intransitive policies~\cite{DBLP:journals/ijisec/BackesP04}. Halevi et al.~\cite{DBLP:journals/iacr/HaleviKN05} address the problem of enforcing confinement in the T10 OSD protocol, in the presence of a fully trusted manager (which has a role similar to the sanitizer in our model). 
Fehr and Fischlin~\cite{FF15} study the case of sanitizable signatures in the context of an intermediate party that sanitizes messages and signatures send over the channel. The special party learns as little as possible about the messages and signatures. However, they do not prevent corrupt senders from sending information to corrupt receivers. 
Finally, the problem of hiding policies and credentials in the context of attribute based encryption has been studied by Frikken et al.~\cite{FAL06}, Kapadia et al.~\cite{KTS07}, M{\"{u}}ller and Katzenbeisser~\cite{DBLP:conf/stm/MullerK11}, and Ferrara et al.~\cite{DBLP:conf/csfw/FerraraFLW15}. However, they do not consider the case of preventing corrupt sender from communicating with corrupt receivers (e.g. by sending the message unencrypted over the channel).


\section{Defining ACE}\label{sec:definition}

\paragraph{ACE Notation.} An \emph{access control encryption} (ACE) scheme is defined by the following algorithms:
\begin{description}
\item[Setup:] The $\Setup$ algorithm on input the security parameter $\kappa$ and a policy $P : [n]\times [n]\to \zo$ outputs a master secret key $msk$ and public parameters $pp$, which include the message space $\M$ and ciphertext spaces $\C,\C'$.\footnote{We use the convention that all other algorithms take $pp$ as input even if not specified. Formally, one can think of the $pp$ as being part of $msk$ and all other keys $ek,dk,sk$.} 

\item[Key Generation:] The $\Gen$ algorithm on input the master secret key $msk$, an identity $i\in \{0,\dots,n+1\}$,\footnote{To make notation more compact we define two special identities: $i=0$ representing a sender or receiver with no rights such that $P(0,j) = 0 = P(i,0)$ for all $i,j \in [n]$; $i=n+1$ to be the sanitizer identity, which cannot receive from anyone but can send to all i.e., $P(n+1,j)=1 \ \forall j\in[n]$ and $P(i,n+1)=0 \ \forall i\in[n]$} and a type $t\in \{\sen,\rec,\san\}$ outputs a key $k$. We use the following notation for the three kind of keys in the system:
	\begin{itemize}
	\item $ek_i\from \Gen(msk,i,\sen)$ and call it an \emph{encryption key for $i\in[n]$} 
	\item $dk_j\from \Gen(msk,j,\rec)$ and call it a \emph{decryption key for $j\in[n]$} 
	\item $ek_0=dk_0=pp$;
	\item $rk \from \Gen(msk,n+1,\san)$ and call it the \emph{sanitizer key};
	\end{itemize}

\item[Encrypt:] The $\Enc$ algorithm on input an encryption key $ek_i$ and a message $m$ outputs a ciphertext $c$. 

\item[Sanitizer:] $\San$ transforms an incoming ciphertext $c\in \C$ into a sanitized ciphertext $c'\in\C'$ using the sanitizer key $rk$;

\item[Decryption:] $\Dec$ recovers a message $m' \in \M \cup \{\bot\}$ from a ciphertext $c'\in\C'$ using a decryption key $dk_j$. 
\end{description}

\paragraph{ACE Requirements.} We formalize Properties 1-3 from the introduction in the following way:

\begin{defi}[Correctness] \label{def:ACE-correct}\label{def:one-correct}
For all $m\in \M$, $i,j \in[n]$ such that $P(i,j)=1$:

$$
\Pr\left[ \Dec\left(dk_j, \San\left(rk, \Enc\left(ek_i ,m\right) \right)\right) \neq m  \right] \leq \negl\left(\kappa\right)
$$
with $(pp,msk)\from \Setup(1^\kappa,P)$, $ek_i\from \Gen(msk,i,\sen)$, $dk_j\from \Gen(msk,j,\rec)$, and $rk \from \Gen(msk,n+1,\san)$, and the probabilities are taken over the random coins of all algorithms.
\end{defi}




\begin{defi}[No-Read Rule] \label{def:oneprivacy}\label{def:ACEnoread} \label{def:privacy} Consider the following game between a challenger $C$ and a stateful adversary $A$:
\begin{center}
\begin{small}
    \begin{tabular}{| p{5.5cm} | p{5cm} |}
    \hline
	\multicolumn{2}{|c|}{\textbf{\emph{No-Read Rule}}} \\
	\hline
	\multicolumn{1}{|c|}{\textbf{\emph{Game Definition}}} & \multicolumn{1}{|c|}{\textbf{\emph{Oracle Definition}}}  \\
	\hline
	\

	1. $(pp,msk)\from \Setup(1^\kappa,P)$;

	2. $(m_0,m_1,i_0,i_1) \from A^{\mathcal{O}_G(\cdot),\mathcal{O}_E(\cdot)}(pp)$;

	3. $b\from \zo$;

	4. $c\from \Enc(\Gen(msk,i_b,\sen),m_b)$;

	5. $b' \from A^{\mathcal{O}_G(\cdot),\mathcal{O}_E(\cdot)}(c)$;

	&
	\

	$\mathcal{O}_G(j,t)$: \newline 1. Output $k\from \Gen(msk,j,t)$;

	\ 

	$\mathcal{O}_E(i,m)$: 

	1. $ek_i\from  \Gen(msk,i,\sen)$;

	2. Output $c\from \Enc( ek_i , m)$;

	\\
	\hline
	\end{tabular}
\end{small}
\end{center}
We say that $A$ \emph{wins the No-Read game} if $b=b'$, $|m_0|=|m_1|$, $i_0,i_1\in\{0,\ldots,n\}$ and one of the following holds:
\begin{description}
	\item[Payload Privacy:] For all queries $q$ to $\mathcal{O}_G$ with $q=(j,\rec)$ it holds that 
		$$P(i_0,j)=P(i_1,j)=0$$
	\item[Sender Anonymity:] For all queries $q$ to $\mathcal{O}_G$ with $q=(j,\rec)$ it holds that 
		$$P(i_0,j)=P(i_1,j) \mbox{ and } m_0=m_1$$
\end{description}
We say an ACE scheme satisfies the \emph{No-Read rule} if for all PPT $A$
$$
\adv{A} = 2\cdot\left|\Pr[A \mbox{ wins the No-Read game}]-\frac{1}{2}\right|\leq \negl(\kappa)
$$ 
\end{defi}

Definition~\ref{def:privacy} captures the requirement that only intended receivers should be able to learn anything about the message (payload privacy) and that no one (even intended receivers) should learn anything about the identity of the sender (sender anonymity). Note that the ciphertext $c$ sent by the challenger to the adversary has \emph{not} been sanitized and that the adversary is allowed to query for the sanitizer key $rk$. This implies that even the sanitizer (even with help of any number of senders and unintended receivers) should not learn anything. Note additionally that the adversary is allowed to query for the encryption keys $ek_{i_0},ek_{i_1}$ corresponding to the challenge identities $i_0,i_1$, which implies that the ability to \emph{encrypt} to a particular identity does not automatically grant the right to \emph{decrypt} ciphertexts created with that identity (e.g., a user might be able to \emph{write} top-secret documents but not to \emph{read} them). Note that if $i_b=0$ for some $b\in\zo$, then the definition implies that it is possible to create ``good looking'' ciphertexts even without having access to any of the senders' keys. This is  explicitly used in our solution with linear complexity. 
Furthermore note that if there exist multiple keys for a single identity (e.g., the output of $\Gen(msk,i,\sen)$ is randomized), then our definition does not guarantee that the adversary can ask the oracle $\mathcal{O}_G$ for the encryption key used to generate the challenge ciphertext. The definition can be easily amended to grant the adversary this power but (since in all our constructions $ek_i$ is a deterministic function of $msk$ and $i$) we prefer to present the simpler definition.
Finally, the encryption oracle $\mathcal{O}_E$ models the situation that the adversary is allowed to see encrypted messages under identities for which he does not have the encryption key.


\begin{defi}[No-Write Rule] \label{def:security}\label{def:ACEnowrite} \label{def:onesecurity}
Consider the following game between a challenger $C$ and a stateful adversary $A$:
\begin{center}
\begin{small}
    \begin{tabular}{| p{6.5cm} | p{5cm} |}
    \hline
	\multicolumn{2}{|c|}{\textbf{\emph{No-Write Rule}}} \\
	\hline
	\multicolumn{1}{|c|}{\textbf{\emph{Game Definition}}} & \multicolumn{1}{|c|}{\textbf{\emph{Oracle Definition}}}  \\
	\hline
	\

	1. $(pp,msk)\from \Setup(1^\kappa,P)$;

	2. $(c,i') \from A^{\mathcal{O}_{E}(\cdot),\mathcal{O}_{S}(\cdot)}(pp)$;

	3. $ek_{i'} \from \Gen(msk,i',\sen)$;

	4. $rk\from \Gen(msk,n+1,\san)$;

	5. $r\from \M$;

	6. $b\from \zo$, 
	\begin{itemize}
		\item If $b=0$, $c'\from \San(rk,\Enc(ek_{i'},r))$; 
		\item If $b=1$, $c'\from \San(rk,c)$;
	\end{itemize}
	7. $b' \from A^{\mathcal{O}_{E}(\cdot),\mathcal{O}_{R}(\cdot)}(c')$;

	&
	\
	
	$\mathcal{O}_S(j,t)$: \newline 1. Output $k\from \Gen(msk,j,t)$;

	\ 

	$\mathcal{O}_R(j,t)$: \newline 1. Output $k\from \Gen(msk,j,t)$;

	\

	$\mathcal{O}_E(i,m)$: 

	1. $ek_i\from  \Gen(msk,i,\sen)$;

	2. $c\from \Enc( ek_i , m)$;
	
	3. Output $c'\from \San(rk,c)$;

	\\
	\hline
	\end{tabular}
\end{small}
\end{center}
Let $Q_S$ (resp. $Q$) be the set of all queries $q=(j,t)$ that $A$ issues to $\mathcal{O}_S$ (resp. both $\mathcal{O}_S$ and $\mathcal{O}_R$). Let $I_S$ be the set of all $i\in[n]$ such that $(i,\sen)\in Q_S$ and let $J$ be the set of all $j\in [n]$ such that $(j,\rec)\in Q$. Then we say that $A$ \emph{wins the No-Write game} if $b'=b$ and all of the following hold:
\begin{enumerate}
	\item $(n+1,\san)\not\in Q$;
	\item $i'\in I_S \cup \{0\}$;
	\item $\forall i\in I_S,j\in J$,  $P(i,j)=0$;
\end{enumerate}
We say an ACE scheme satisfies the \emph{No-Write rule} if for all PPT $A$
$$
\adv{A} = 2\cdot \left|\Pr[A \mbox{ wins the No-Write game}]-\frac{1}{2}\right|\leq \negl(\kappa)
$$ 
\end{defi}

Definition~\ref{def:security} captures the property that any set of (corrupt) senders $\{S_i\}_{i\in I}$ cannot transfer any information to any set of (corrupt) receivers $\{R_j\}_{j\in J}$ unless at least one of the senders in $I$ is allowed communication to at least one of the receivers in $J$ (Condition 3)\footnote{Note that the adversary is allowed to ask for any senders' key in the post-challenge queries.}. This is modelled by saying that in the eyes of the receivers, the sanitized version of a ciphertext coming from this set of senders looks like the sanitized version of a fresh encryption of a random value produced by one of these senders (Condition 2). Note that if the adversary does not ask for any encryption key (i.e., $I_S=\emptyset$), then the only valid choice for $i'$ is $0$: this implies that (as described for the no-read rule) there must be a way of constructing ``good looking'' ciphertexts using the public parameters only\footnote{Recall that we defined $ek_0=pp$.} and this property is used crucially in the construction of our linear scheme. 
Furthermore, we require that the adversary does not corrupt the sanitizer (Condition 1) which is, as discussed in the introduction, an unavoidable condition. 
Finally, the encryption oracle $\mathcal{O}_E$ again models the situation that the adversary is allowed to see encrypted messages under identities for which he do not have the encryption key.


\section{Linear ACE from Standard Assumptions}\label{sec:linear}

The roadmap of this section is the following: we construct an ACE scheme for a single identity (i.e., $n=1$ and $P(1,1)=1$) from standard number theoretic assumptions, and then we construct an ACE scheme for any predicate $P:[n]\times [n]\to \zo$ using a \emph{repetition scheme}. The complexity of the final scheme (in terms of public-key and ciphertext size) is $n$ times the complexity of the single-identity scheme.


\subsection{ACE for a Single Identity}

We propose two constructions of ACE for a single identity (or \oACE for short). The first is based on the DDH assumption and is presented in this section, while the second is based on the security of Pailler's cryptosystem and is deferred to Appendix~\ref{app:onepailler}. 
Both schemes share the same basic idea: the encryption key $ek$ is some secret value $\alpha$, and an encryption of a message $m$ is a pair of encryptions $(c_0,c_1)=(E(\alpha),E(m))$. The sanitizer key is also the value $\alpha$, and a sanitized ciphertext is computed as $c'=c_1\cdot (c_0\cdot E(-\alpha))^s$ which (thanks to the homomorphic properties of both ElGamal and Pailler) is an encryption of a uniformly random value unless $c_0$ is an encryption of $\alpha$, in which case it is an encryption of the original message $m$. The decryption key is simply the decryption key for the original encryption scheme, which allows to retrieve $m$ from $c'$. Note that even knowing the decryption key is not enough to construct ciphertexts which ``resist'' the sanitization, since the receiver never learns the value $\alpha$.

\paragraph{\oACE from DDH:}
Our first instantiation is based on the ElGamal public-key encryption scheme \cite{DBLP:journals/tit/Elgamal85}. The construction looks similar to other \emph{double-strand} versions of ElGamal encryption which have been used before in the literature to achieve different goals (e.g., by Golle et al.~\cite{DBLP:conf/ctrsa/GolleJJS04} in the context of universal re-encryption and by Prabhakaran and Rosulek~\cite{DBLP:conf/crypto/PrabhakaranR07} in the context of rerandomizable CCA security).

\begin{construction} \label{con:oneElGamal}
Let $\EGACE = (\Setup, \Gen, \Enc, \San, \Dec)$ be a \oACE scheme defined by the following algorithms:
\end{construction}
\begin{description}
\item[Setup:] Let $(G,q,g)$ be the description of a cyclic group of prime order $q$ generated by $g$. Let $(\alpha, x) \from \mathbb{Z}_q\times \mathbb{Z}_q$ be uniform random elements, and compute $h=g^x$. Output the public parameter  $pp = (G,q,g,h)$ and the master secret key $msk = (\alpha,x)$. The message space is $\M=G$ and the ciphertext spaces are $\C = G^4$ and $\C' = G \times G$.

\item[Key Generation:] Given the master secret key $msk$, the encryption, decryption and sanitizer key are computed as follows:
	\begin{itemize}
	\item $ek = \alpha$;
	\item $dk = - x$;
	\item $rk = -\alpha$;
	\end{itemize}

\item[Encryption:] Given the message $m$ and an encryption key $ek$, sample random $r_1,r_2 \in \mathbb{Z}_q$ and output:
\begin{align*}
	(c_0, c_1, c_2, c_3) = (g^{r_1}, g^{ek} h^{r_1}, g^{r_2}, mh^{r_2})
\end{align*}
(and encryptions for the identity $0$ are random tuples from $G^4$).

\item[Sanitize:] Given a ciphertext $c = (c_0, c_1, c_2, c_3) \in \C$ and a sanitizer key $rk$, sample uniform random $s_1,s_2 \in \mathbb{Z}_q$ and output:
\begin{align*}
	(c'_0, c'_1 )= (c_2c_0^{s_1}g^{s_2},c_3 (g^{rk} c_1)^{s_1} h^{s_2})
\end{align*}

\item[Decryption:] Given a ciphertext $c' = (c'_0, c'_1) \in \C'$ and a decryption key $dk$, output:
\begin{align*}
	m' = c'_1(c'_0)^{dk}
\end{align*}
\end{description}


\begin{lem}\label{lem:elgamal}
Construction~\ref{con:oneElGamal} is a \emph{correct} \oACE scheme that satisfies the \emph{No-Read Rule} and the the \emph{No-Write Rule} assuming that the DDH assumption holds in $G$.
\end{lem}

\begin{proof}
\emph{Correctness:} 
Let $c = (c_0, c_1, c_2, c_3)$  be an honestly generated ciphertext, and let $c' = (c'_0, c'_1)$ be a sanitized version of $c$. We check that $(c'_0,c'_1)$ is still an encryption of the original message $m$:
\begin{align*}
	c'_1(c'_0)^{dk} &= c_3 (g^{rk} c_1)^{s_1} h^{s_2}(c_2c_0^{s_1}g^{s_2})^{dk} \\
&= 	mh^{r_2} (g^{-\alpha} g^\alpha h^{r_1})^{s_1}h^{s_2}(g^{r_2} g^{r_1s_1+s_2})^{-x} \\
&= m h^{r_2+r_1s_1+s_2} g^{-x(r_2+r_1s_1+s_2)} = m
\end{align*}
Thus, the sanitization of a valid ciphertext produces a new valid ciphertext under the same identity and of the same message.
\\

\noindent \emph{No-Read Rule:}
There are three possible cases, depending on which identities the adversary queries during the game: the case $(i_0,i_1)=(0,0)$ is trivial as both $\Enc(ek_0,m_b)$ for $b\in\zo$ are random ciphertexts; the case $(i_0,i_1)=(1,1)$ is trivial if the adversary asks for the decryption key $dk$, since in this case it must be that $m_0=m_1$. The case where the adversary does not ask for the decryption key and $i_0\neq i_1$ implies the case where the adversary does not ask for the decryption key and $(i_0,i_1)=(1,1)$ using standard hybrid arguments (i.e., if $\Enc(ek,m)$ is indistinguishable from a random ciphertext $c\from \C$ for all $m$, then $\Enc(ek,m_0)$ is indistinguishable from $\Enc(ek,m_1)$ for all $m_0,m_1$). So we are only left to prove that honest encryptions are indistinguishable from a random element in $\C=G^4$, which follows in a straightforward way from the DDH assumption. In particular, since $(g,h,g^{r_2},h^{r_2})$ is indistinguishable from $(g,h,g^{r_2},h^{r_3})$ for random $r_2,r_3$ we can replace $c_3$ with a uniformly random element (independent of $m$). Notice that neither $\alpha$ (the encryption and sanitizer key) nor encryptions from oracle $\mathcal{O}_E$ will help the adversary distinguish. 
Thus, we can conclude that the adversary's advantage is negligible, since he cannot distinguish in all three cases.
\\

\noindent \emph{No-Write Rule:}
We only need to consider two cases, depending on which keys the adversary asks for \emph{before} producing the challenge ciphertext $c$ and identity $i'$: 1) the adversary asks for $ek$ before issuing his challenge $(c,i')$ with $i'\in\zo$ (and receives no more keys during the distinguishing phase) and 2) the adversary asks for $dk$ before issuing his challenge $(c,0)$ and then asks for $ek$ during the distinguishing phase. Case 1) follows directly from the DDH assumption: without access to the decryption key the output of the sanitizer is indistinguishable from a random ciphertext thanks to the choice of the random $s_2$, in particular since $(g,h,g^{s_2},h^{s_2})$ is indistinguishable from $(g,h,g^{s_2},h^{s_3})$ for random $s_2,s_3$ we can replace $(c'_0,c'_1)$ with uniformly random elements in $G$. Case 2) instead has to hold unconditionally, since the adversary has the decryption key. We argue that the distribution of $\San(rk,(c_0,c_1,c_2,c_3))$ is independent of its input. In particular, given any (adversarially chosen) $(c_0,c_1,c_2,c_3)\in  G^4 $ we can write:
$$
(c_0,c_1,c_2,c_3)=(g^{\delta_0},g^{\delta_1},g^{\delta_2} ,g^{\delta_3}  )
$$
Then the output $c'\from \San(rk,c)$ is
\begin{align*}
(c'_0,c'_1) &=  (c_2c_0^{s_1}g^{s_2},c_3 (g^{rk} c_1)^{s_1} h^{s_2}) \\
&= (g^{\delta_2+s_1\delta_0+s_2},g^{\delta_3+s_1(\delta_1-\alpha)+s_2x} )
\end{align*}
Which is distributed exactly as a uniformly random ciphertext $(g^{\gamma_0},g^{\gamma_1})$ with $(\gamma_0,\gamma_1)\in  \mathbb{Z}_{q}\times\mathbb{Z}_{q}$ since for all $(\gamma_0,\gamma_1)$ there exists $(s_0,s_1)$ such that:
$$
\gamma_0 = \delta_2+s_1\delta_0+s_2 \mbox{ and } \gamma_1=\delta_3+s_1(\delta_1-\alpha)+s_2x
$$
This is guaranteed unless the two equations are linearly dependent i.e., unless $\alpha=(\delta_1-x\delta_0)$ which happens only with negligible probability thanks to the principle of deferred decisions.
The adversary is allowed to see sanitized ciphertext from the encryption oracle. However, this does not help him distinguish, since the output of the $\San$ algorithm is distributed exactly as a uniform ciphertext.
Thus, we can conclude that the adversary's advantage is negligible, since he cannot distinguish in both cases.
\end{proof}

\subsection{Construction of an ACE Scheme for Multiple Identities}

In this section we  present a construction of an ACE scheme for multiple identities, which is based on the \oACE scheme in a black-box manner. In a nutshell, the idea is the following: we run $n$ copies of the \oACE scheme and we give to each receiver $j$ the decryption key $dk_j$ for the $j$-th copy of the scheme. An encryption key for identity $i$ is given by the set of encryption keys $ek_j$ of the \oACE scheme such that $P(i,j)=1$. To encrypt, a sender encrypts the same message $m$ under all its encryption keys $ek_j$ and puts random ciphertexts in the positions for which he does not know an encryption key. The sanitizer key contains all the sanitizer keys for the \oACE scheme: this allows the sanitizer to sanitize each component independently, in such a way that for all the positions for which the sender knows the encryption key, the message ``survives'' the sanitization, whereas in the other positions the output is uniformly random.

\paragraph{Example.} We conclude this informal introduction of our repetition scheme by giving a concrete example of an ACE scheme for the Bell-LaPadula access control policy with three levels of access: \textit{Level 1: top-secret, level 2: secret, level 3: public}.
The predicate of this access control is defined as $P(i,j) = 1 \Leftrightarrow i \geq j$. This predicate ensures the property of \textit{no write down} and \textit{no read up} as discussed in the introduction.
Table~\ref{table:BL_ACE} shows the structure of the keys and the ciphertext for the different levels of access.

\renewcommand{\oneACE}{}
\renewcommand{\oneACEpow}{{}}
\begin{table}
\centering
\begin{tabular}{l|l|ccc|ccccc}
$i$		&	\multicolumn{1}{|c|}{$ek_i$}	
		&	\multicolumn{3}{|c}{$dk_i$}	
		&	\multicolumn{5}{|c}{$c$}		\\ \hline
		
1		&	$\{ek_1^\oneACEpow \}$						
		&	$dk_1^\oneACEpow$	
		&&&	$($&$ \Enc(ek_1^\oneACEpow, m),$&$c'_2,$&$ c'_3 $&$)$ \\
		
2		&	$\{ek_1^\oneACEpow, ek_2^\oneACEpow \}$				
		&&	$dk_2^\oneACEpow$	
		&&	$($&$ \Enc(ek_1^\oneACEpow, m),$&$ \Enc(ek_2^\oneACEpow, m),$&$ c''_3 $&$)$ \\
		
3		&	$\{ek_1^\oneACEpow, ek_2^\oneACEpow, ek_3^\oneACEpow \}$			
		&&&	$dk_3^\oneACEpow$	
		&	$($&$ \Enc(ek_1^\oneACEpow, m),$&$ \Enc(ek_2^\oneACEpow, m),$&$ \Enc(ek_3^\oneACEpow, m) $&$)$ \\
\end{tabular}
\newline
\caption{Access Control Encryption Scheme for Bell-LaPadula access control policy. $c'_2,c'_3,c''_3$ are random ciphertexts from $\C$.}
\label{table:BL_ACE}
\end{table}
\renewcommand{\oneACE}{\mathsf{1ACE}}
\renewcommand{\oneACEpow}{\mathsf{1ACE}}

\begin{construction} \label{con:ACE}
Let $\oneACE = (\Setup, \Gen, \Enc, \San, \Dec)$ be a \oACE scheme. Then we can construct an ACE scheme $\ACE = (\Setup, \Gen, \Enc, \San, \Dec)$ defined by the following algorithms:
\end{construction}
\begin{description}
\item[Setup:] Let $n$ be the number of senders/receivers specified by the policy $P$. Then run $n$ copies of the $\oneACE$ setup algorithm
\begin{align*}
	(pp_i^\oneACEpow,msk_i^\oneACEpow) \gets \oneACE.\Setup(1^\secparam) \quad \text{for } i = 1,\dots,n
\end{align*}
For each of the $\oneACE$ master secret keys run the $\oneACE$ key generation algorithm on each of the three modes. For $i\in [n]$ do the following
\begin{align*}
	ek_i^\oneACEpow &\gets \oneACE.\Gen(msk_i^\oneACEpow, \sen) \\
	dk_i^\oneACEpow &\gets \oneACE.\Gen(msk_i^\oneACEpow, \rec) \\
	rk_i^\oneACEpow &\gets \oneACE.\Gen(msk_i^\oneACEpow, \san) 
\end{align*}

Output the public parameter and the master secret key\footnote{There exists some encoding function that takes a message $m$ from the message space of the $\ACE$ scheme and encodes it into a message of each of the \oACE message spaces. The ciphertext spaces of the $\ACE$ scheme are the crossproduct of all the \oACE ciphertext spaces, thus $\C = \C_1^\oneACEpow \times \cdots \times \C_n^\oneACEpow$ and $\C' = {\C'_1}^\oneACEpow \times \cdots \times {\C'_n}^\oneACEpow$.}
$$
	pp = \{ pp_i^\oneACEpow \}_{i \in [n]}, \qquad msk := \{ ek_i^\oneACEpow,dk_i^\oneACEpow, rk_i^\oneACEpow \}_{i \in [n]}
$$

\item[Key Generation:] On input an identity $i\in \{0,\dots,n+1\}$, a mode $\{\sen,\rec,\san\}$ and the master secret key $msk$, output a key depending on the mode
	\begin{itemize}
	\item $ek_i := \{ek_j^\oneACEpow \}_{j \in S}$, where $S \subseteq [n]$ is the subset s.t. $j \in S$ iff $P(i,j) = 1$;
	
	\item $dk_i := dk_i^\oneACEpow$;
	
	\item $rk := \{ rk_j^\oneACEpow \}_{j \in [n]}$;
	\end{itemize}

\item[Encrypt:] On input an encryption key $ek_i$ and a message $m$ encrypt the message under each of the $\oneACE$ encryption keys in $ek_i$ and sample uniform random ciphertext for each public key not in the encryption key. 
Thus, for $j = 1,\dots,n$ do the following
\begin{itemize}
\item If $ek_j^\oneACEpow \in ek_i$ then compute $c_j^\oneACEpow \gets \oneACE.\Enc(ek_j^\oneACEpow, m)$.

\item If $ek_j^\oneACEpow \notin ek_i$ then sample $c_j^\oneACEpow \gets_\$ \C_j^\oneACEpow$.\footnote{Here $c_j^\oneACEpow \gets_\$ \C_j^\oneACEpow$ is a shorthand for $c_j^\oneACEpow \gets \oneACE.\Enc(pp_j^\oneACEpow, \bot)$.}
\end{itemize}
Output the ciphertext $c := \left( c_1^\oneACEpow, \dots, c_n^\oneACEpow  \right)$.

\item[Sanitizer:] On input a ciphertext $c$ and a sanitizer key $rk$, sanitize each of the $n$ $\oneACE$ ciphertexts as follows
\begin{align*}
	{c'_i}^\oneACEpow \gets \oneACE.\San(rk_i^\oneACEpow,c_i^\oneACEpow) \quad \text{for } i = 1,\dots,n
\end{align*} 
Output the sanitized ciphertext $c' := ({c'_1}^\oneACEpow,\dots,{c'_n}^\oneACEpow)$.

\item[Decryption:] On input a ciphertext $c$ and a decryption key $dk_i$ decrypt the $i$'th $\oneACE$ ciphertext
\begin{align*}
	m' \gets \oneACE.\Dec(dk_i^\oneACEpow,c_i^\oneACEpow)
\end{align*}
\end{description}

We can prove that the scheme presented above satisfies \emph{correctness} as well as the \emph{no-read} and the \emph{no-write} rule, by reducing the properties of the repetition scheme to properties of the scheme with a single identity using hybrid arguments. 
The formal proofs are deferred to Appendix~\ref{app:repetitionproofs}. 

\section{Polylogarithmic ACE from iO}\label{sec:polylog}

In this section, we present our construction of ACE with polylogarithmic complexity in the number of roles $n$.

At first it might seem that it is easy to construct an ACE scheme with short ciphertexts by using NIZK and re-randomizable encryption: the sender would send to the sanitizer a ciphertext and a NIZK proving that the ciphertext is a well-formed encryption of some message using a public key that the sender is allowed to send to (for instance, each sender could have a signature on their identity to be able to prove this statement). Now the sanitizer drops the NIZK and passes on the re-randomized ciphertext. However, the problem is that the sanitizer would need to know the public key of the intended receiver to be able to re-randomize (and we do not want to reveal who the receiver is). 

As described in the introduction, we build our ACE scheme on top of a FE scheme which is \emph{sanitizable}, which roughly means that given a ciphertext it is possible to produce a new encryption of the same message which is independent of the randomness used in the original encryption (this is a relaxation of the well-known \emph{re-randomizability} property, in the sense that we do not require sanitized ciphertexts to look indistinguishable from fresh encryptions e.g., they can be syntactically different). We construct such an FE scheme by modifying the FE scheme of Garg et al.~\cite{DBLP:conf/focs/GargGH0SW13}, and therefore our construction relies on the assumption that indistinguishability obfuscation exists. We define and construct sFE in Section~\ref{sec:sfe} and then construct ACE based on sFE (and a regular PRF) in Section~\ref{sec:polylogace}.

\subsection{Sanitizable Functional Encryption Scheme -- Definition}\label{sec:sfe}
A \emph{sanitizable functional encryption} (sFE) scheme is defined by the following algorithms:
\begin{description}
\item[Setup:] The $\Setup$ algorithm on input the security parameter $\kappa$ outputs a master secret key $msk$ and public parameters $pp$, which include the message space $\M$ and ciphertext spaces $\C,\C'$.

\item[Key Generation:] The $\Gen$ algorithm on input the master secret key $msk$ and a function $f$, outputs a corresponding secret key $SK_f$.

\item[Encrypt:] The $\Enc$ algorithm on input the public parameters $pp$ and a message $m$, outputs a ciphertext $c \in \C$

\item[Sanitizer:] The $\San$ algorithm on input the public parameters $pp$ and a ciphertext $c \in \C$, transforms the incoming ciphertext into a sanitized ciphertext $c'\in\C'$ 

\item[Decryption:] The $\Dec$ algorithm on input a secret key $SK_f$ and a sanitized ciphertext $c'\in\C'$ that encrypts message $m$, outputs $f(m)$.



\end{description}

For the sake of exposition we also define a \emph{master decryption algorithm} that on input $c\from\Enc(pp,m)$, returns $m \from \MasterDec(msk,c)$.\footnote{Formally $\MasterDec$ is a shortcut for $\Dec(\Gen(msk,f_{id}),\San(pp,c))$, where $f_{id}$ is the identity function.} We formally define \emph{correctness} and \emph{IND-CPA} security for an sFE scheme (which are essentially the same as for regular FE), and then we define the new \emph{sanitizable} property which, as described above, is a relaxed notion of the re-randomization property.

\begin{defi}[Correctness for sFE]
Given a function family $\mathcal{F}$. For all $f \in \mathcal{F}$ and all messages $m \in \M$:
	\begin{align*}
	\Pr \left[ \Dec( \Gen(msk,f), \San(pp,\Enc(pp,m))) \neq f(m) \right] \leq \negl(\kappa)
	\end{align*}
where $(pp,msk) \from \Setup(1^\kappa)$ and the probabilities are taken over the random coins of all algorithms.
\end{defi}

\begin{defi}[IND-CPA Security for sFE] \label{def:FEnoread}
Consider the following game between a challenger $C$ and a stateful adversary $A$:
\begin{center}
\begin{small}
    \begin{tabular}{| p{5.5cm} | p{5cm} |}
    \hline
	\multicolumn{2}{|c|}{\textbf{\emph{ IND-CPA Security}}} \\
	\hline
	\multicolumn{1}{|c|}{\textbf{\emph{Game Definition}}} & \multicolumn{1}{|c|}{\textbf{\emph{Oracle Definition}}}  \\
	\hline
	\
	
	1. $(pp,msk)\from \Setup(1^\kappa)$;

	3. $(m_0,m_1) \from A^{\mathcal{O}(\cdot)}(pp)$;
	
	4. $b\from \zo$;

	5. $c^* \from \Enc(pp,m_b)$
	
	6. $b' \from A^{\mathcal{O}(\cdot)}(c^*)$;

	&
	\

	$\mathcal{O}(f_i)$: \newline 1. Output $SK_{f_i} \from \Gen(msk,f_i)$;

	\\
	\hline
	\end{tabular}
\end{small}
\end{center}
We say that $A$ \emph{wins the IND-CPA game} if $b=b'$, $|m_0|=|m_1|$, and that $f_i(m_0) = f_i(m_1)$ for all oracle queries $f_i$.
We say a sFE scheme satisfies the IND-CPA security property if for all PPT $A$
$$
\adv{A} = 2\cdot \left|\Pr[A \mbox{ wins the IND-CPA game}]-\frac{1}{2}\right|\leq \negl(\kappa)
$$ 
\end{defi}

\begin{defi}[Sanitization for sFE] \label{def:FEnowrite}
Consider the following game between a challenger $C$ and a stateful adversary $A$:
\begin{center}
\begin{small}
    \begin{tabular}{| p{8cm} |}
    \hline
	\multicolumn{1}{|c|}{\textbf{\emph{Sanitization}}} \\
	\hline
	\multicolumn{1}{|c|}{\textbf{\emph{Game Definition}}}  \\
	\hline
	\

	1. $(pp,msk)\from \Setup(1^\kappa)$;

	2. $c \from A(pp,msk)$;

	3. $b\from \zo$, 
	\begin{itemize}
		\item If $b=0$, $c^*\from \San(pp,c)$;
		\item If $b=1$, $c^*\from \San(pp,\Enc(pp,\MasterDec(msk,c)))$; 
	\end{itemize}
	4. $b' \from A(c^*)$;

	\\
	\hline
	\end{tabular}
\end{small}
\end{center}
We say that $A$ \emph{wins the sanitizer game} if $b=b'$.
We say a sFE scheme is sanitizable if for all PPT $A$
$$
\adv{A} = 2\cdot \left|\Pr[A \mbox{ wins the sanitizer game}]-\frac{1}{2}\right|\leq \negl(\kappa)
$$ 
\end{defi}

Note that in Definition~\ref{def:FEnowrite} the adversary has access to the master secret key.






\subsection{Sanitizable Functional Encryption Scheme -- Construction}\label{sec:sfecons}
We now present a construction of a sFE scheme based on iO. The construction is based on the functional encryption construction by Garg et. al~\cite{DBLP:conf/focs/GargGH0SW13}. In their scheme a ciphertext contains two encryptions of the same message and a NIZK of this statement, thus an adversary can leak information via the randomness in the encryptions or the randomness in the NIZK. In a nutshell we make their construct sanitizable by:
\begin{enumerate}
\item Replacing the PKE scheme with a \emph{sanitizable PKE} (as formalized in Definition~\ref{def:PKE_sanitation}).
\item Letting the sanitizer drop the original NIZK, and append a \emph{proof of a proof} instead (i.e., a proof that the sanitizer knows a proof that would make the original verifier accept). Thanks to the ZK property the new NIZK does not contain any information about the randomness used to generate the original NIZK.
\item Changing the decryption keys (obfuscated programs) to check the new proof instead.

\end{enumerate}

\paragraph{Building Blocks.} We formalize here the definition of sanitization for a PKE scheme. Any re-randomizable scheme (such as Paillier and ElGamal) satisfies perfect PKE sanitization, but it might be possible that more schemes fit the definition as well.

\begin{defi}[Perfect PKE Sanitization] \label{def:PKE_sanitation}
Let $\M$ be the message space and $\R$ be the space from which the randomness for the encryption and sanitization is taken. 
Then for every message $m \in \M$ and for all $r,s,r' \in \R$ there exists $s' \in \R$ such that
\begin{align*}
	\San(pk,\Enc(pk,m;r);s) = \San(pk,\Enc(pk,m;r');s')
\end{align*}
\end{defi}

Our constructions also uses (by now standard) tools such as \emph{pseudo-random functions (PRF)}, \emph{indistinguishability obfuscation (iO)} and \emph{statistical simulation-sound non-interactive zero-knowledge (SSS-NIZK)}, which are defined for completeness in Appendix~\ref{app:bb}.

\paragraph{Constructing sFE.} We are now ready to present our construction of sFE.

\begin{construction} \label{con:sanFE}
Let $\sanPKE = (\Setup, \Enc, \San, \Dec)$ be a perfect sanitizable public key encryption scheme.
Let $\NIZK = (\Setup, \Prove, \Verify)$ be a statistical simulation-sound NIZK.
Let $iO$ be an indistinguishability obfuscator.
We construct a sanitizable functional encryption scheme $\sanFE = (\Setup, \Gen, \Enc,\San, \Dec)$ as follows:
\end{construction}
\begin{description}
\item[Setup:] On input the security parameter $\kappa$ the setup algorithm compute the following
	\begin{enumerate}
	\item $(pk_1,sk_1) \gets \sanPKE.\Setup(1^\kappa)$;
	\item $(pk_2,sk_2) \gets \sanPKE.\Setup(1^\kappa)$;
	\item $crs_E \gets \NIZK.\Setup(1^\kappa,R_E)$;
	\item $crs_S \gets \NIZK.\Setup(1^\kappa,R_S)$;
	\item Output  $pp = (crs_E, crs_S, pk_1, pk_2)$ and  $msk = sk_1$;
	\end{enumerate}

	The relations $R_E$ and $R_S$ are defined as follows: Let $x_E=(c_1,c_2)$ be a statement and $w_E = (m,r_1,r_2)$ a witness, then $R_E$ is defined as
	\begin{align*}
		R_E = \left\lbrace (x_E,w_E) \mid c_1 = \sanPKE.\Enc(pk_1,m;r_1) \wedge c_2 = \sanPKE.\Enc(pk_2,m;r_2) \right\rbrace
	\end{align*}
	
	Let $x_S=(c'_1,c'_2)$ be a statement and $w_S = (c_1,c_2,s_1,s_2,\pi_E)$ a witness, then $R_S$ is defined as
	\begin{align*}
		R_S = \left\lbrace (x_S,w_S) \middle|	
			\begin{array}{l}
				c'_1 = \sanPKE.\San(pk_1,c_1;s_1) \wedge c'_2 = \sanPKE.\Enc(pk_2,c_2;s_2)  \\
				\wedge \NIZK.\Verify(crs_E,(c_1,c_2),\pi_E) = 1
			\end{array} 
		\right\rbrace
	\end{align*}

\item[Key Generation:] On input the master secret key $msk$ and a function $f$ output the secret key $SK_f = iO(P)$ as the obfuscation of the following program

\begin{center}
\begin{small}
    \begin{tabular}{| p{6cm} |}
	\hline
	\multicolumn{1}{|c|}{\textbf{\emph{Program} $P$}} \\
	\hline 
	\
	
	Input: $c'_1,c'_2,\pi_S$; 
	
	Const: $crs_S,f,sk_1$;
	
	\

	1. If $\NIZK.\Verify(crs_S,(c'_1,c'_2),\pi_S) = 1$;
	\begin{itemize}
		\item[] output $f(\sanPKE.\Dec(sk_1,c'_1))$;
	\end{itemize}

	2. else output fail;

	\\
	\hline
	\end{tabular}
\end{small}
\end{center}

\item[Encrypt:] On input the public parameters $pp$ and a message $m$ compute two PKE encryptions of the message
	\begin{align*}
		c_1 &\gets \sanPKE.\Enc(pk_1,m;r_1) \\
		c_2 &\gets \sanPKE.\Enc(pk_2,m;r_2)
	\end{align*}
	with randomness $(r_1,r_2)$.	Then create a proof $\pi_E$ that $(x_E,w_E) \in R_E$ with $x_E = (c_1,c_2)$ and witness $w_E = (m,r_1,r_2)$ 
	\begin{align*}
		\pi_E \gets \NIZK.\Prove(crs_E,x_E,w_E;t_E)
	\end{align*}
	using randomness $t_E$. 
	Output the triple $c = (c_1,c_2,\pi_E)$ as the ciphertext.

\item[Sanitizer:] On input the public parameter $pp$ and a ciphertext $c = (c_1,c_2,\pi_E) \in \C$ compute the following
	\begin{enumerate}
	\item If $\NIZK.\Verify(crs_E,x_E,\pi_E) = 1$ then
		\begin{itemize}
		\item[] $c'_1 \gets \sanPKE.\San(pk_1,c_1;s_1)$
		\item[] $c'_2 \gets \sanPKE.\San(pk_2,c_2;s_2)$
		\item[] $\pi_2 \gets \NIZK.\Prove(crs_S,x_S,w_S;t_S)$
		\item[] Output $c' = (c'_1,c'_2,\pi_2)$
		\end{itemize}
	\item Else
		\begin{itemize}
		\item[] Output $c' \from \sanFE.\San(pp,\sanFE.\Enc(pp,\bot))$
		\end{itemize}
	\end{enumerate}
	with randomness $(s_1,s_2)$ and $t_S$ in the PKE and NIZK respectively. The  generated NIZK is a proof that $(x_S,w_S) \in R_S$ with $x_S = (c'_1,c'_2)$  and $w_S = (c_1,c_2,s_1,s_2,\pi_E)$.

\item[Decryption:] On input a secret key $SK_f$ and a ciphertext $c' = (c'_1,c'_2,\pi_S) \in \C'$, run the obfuscated program $SK_f(c'_1,c'_2,\pi_S)$ and output the answer.

\end{description}

\begin{lem}
Construction~\ref{con:sanFE} is a correct functional encryption scheme.
\end{lem}

\begin{proof}
Correctness follows from the correctness of the $iO$, PKE, and SSS-NIZK schemes, and from inspection of the algorithms.
\end{proof}

\begin{lem}\label{lem:IND-CPA_sFE}
For any adversary $A$ that breaks the IND-CPA security property of Construction~\ref{con:sanFE}, there exists an adversary $B$ for the computational zero-knowledge property of the NIZK scheme, an adversary $C$ for the IND-CPA security of the PKE scheme, and an adversary $D$ for iO such that the advantage of adversary $A$ is
\begin{align*}
	\adv{\sanFE,A} \leq 4|\M| \left( \adv{\NIZK,B} + \adv{\sanPKE,C} + q\cdot\adv{iO,C} (1 - 2\sssprob) \right)
\end{align*}
where $q$ is the number of secret key queries adversary $A$ makes during the game, and $\sssprob$ is the negligible soundness error of the SSS-NIZK scheme.
\end{lem}

\begin{proof}
This proof follows closely the selective IND-CPA security proof of the FE construction presented by Garg et. al.~\cite{DBLP:conf/focs/GargGH0SW13}.
See Appendix~\ref{app:IND-CPA_sFE} for the full proof.
\end{proof}

\begin{lem}\label{lem:Sanitize_sFE}
For any adversary $A$ that breaks the sanitizer property of Construction~\ref{con:sanFE}, there exists an adversary $B$ for the computational zero-knowledge property of the NIZK scheme such that the advantage of adversary $A$ is
\begin{align*}
	\adv{\sanFE,A} \leq 2|\M| \adv{\NIZK,B} 
\end{align*}
\end{lem}

\begin{proof}
This lemma is proven via a series of indistinguishable hybrid games between the challenger and the adversary. For the proof to go through we notice that the challenger needs to simulate the NIZK proof. At a first look it might seem that the reduction needs to guess the entire ciphertext before setting up the system parameter, but in fact we show that it is enough to guess the \emph{message} beforehand! Thus, we can use a complexity leveraging technique to get the above advantage.
See Appendix~\ref{app:Sanitize_sFE} for the full proof.

\end{proof}

\subsection{Polylog ACE scheme}\label{sec:polylogace}

In this section, we present a construction of an ACE scheme for multiple identities based on sanitizable functional encryption. The idea of the construction is the following: an encryption of a message $m$ is a sFE encryption of the message together with the senders identity $i$ and a MAC of the message based on the identity. Crucially, the encryption keys for all identities are generated in a pseudorandom way from a master key, thus it is possible to check MACs for all identities using a compact circuit.
The sanitizer key is a sFE secret key for a special function that checks that the MAC is correct for the claimed identity. 
Then the sanitization consists of sanitizing the sFE ciphertext, and then using the sanitizer key to check the MAC.
The decryption key for identity $j$ is a sFE secret key for a function that checks that identity $i$ in the ciphertext and identity $j$ are allowed to communicate (and ignores the MAC). The function then outputs the message iff the check goes through.

\begin{construction} \label{con:ACEfromFE}
Let $\sanFE = (\Setup, \Gen, \Enc, \San, \Dec)$ be a sanitizable functional encryption scheme. Let $F_1,F_2$ be pseudorandom functions.
Then we can construct an ACE scheme $\ACE = (\Setup, \Gen, \Enc, \San, \Dec)$ defined by the following algorithms:
\end{construction}
\begin{description}
\item[Setup:]
Let $K \from \zo^\kappa$ be a key for the pseudorandom function $F_1$. Run $(pp^\sanFEpow,msk^\sanFEpow) \from \sanFE.\Setup(1^\kappa)$. 
Output the public parameter $pp = pp^\sanFEpow$ and the master secret key $msk = (msk^\sanFEpow,K)$

\item[Key Generation:]
Given the master secret key $msk$ and an identity $i$, the encryption, decryption and sanitizer key are computed as follows:
\begin{itemize}
\item $ek_i \from F_1(K,i)$
\item $dk_i \from \sanFE.\Gen(msk^\sanFEpow, f_i)$
\item $rk \from \sanFE.\Gen(msk^\sanFEpow, f_{rk})$
\end{itemize}
where the functions $f_i$ and $f_{rk}$ are defined as follows

\begin{center}
\begin{small}
    \begin{tabular}{| p{5.5cm} | p{5cm} |}
	\hline
	\multicolumn{1}{|c|}{\textbf{\emph{Decryption function}}} & \multicolumn{1}{|c|}{\textbf{\emph{Sanitizer function}}}  \\
	\hline
	\
	
	$f_i(m,j,t)$:

	1. If $P(j,i) = 1$: output $m$;

	2. Else output $\bot$;

	&
	\

	$f_{rk}(m,j,t)$: 
	
	1. $ek_j = F_{1}(K,j)$;
	
	2. If $t = F_{2}(ek_j,m)$: output $1$;
	
	3. Else output $0$;

	\\
	\hline
	\end{tabular}
\end{small}
\end{center}

\item[Encryption:]
On input a message $m$ and an encryption key $ek_i$, compute $t = F_2(ek_i,m)$ and output $$c = \sanFE.\Enc(pp^\sanFEpow,(m,i,t))$$

\item[Sanitizer:]
Given a ciphertext $c$ and the sanitizer key $rk = SK_{rk}$ check the MAC and output a sanitized FE ciphertext
\begin{enumerate}
\item $c' = \sanFE.\San(pp^\sanFEpow,c)$
\item If $\sanFE.\Dec(SK_{rk},c') = 1$: output $c'$
\item Else output $\San(rk,\Enc(ek_0,\bot))$ 
\end{enumerate}

\item[Decryption:] 
Given a ciphertext $c'$ and a decryption key $dk_j = SK_j$ output $$m' = \sanFE.\Dec(SK_j,c')$$
\end{description}

\begin{lem} \label{lem:FEcorrect}
Construction~\ref{con:ACEfromFE} is a \emph{correct} ACE scheme 
\end{lem}

\begin{proof}
Let $P(i,j) = 1$ for some $i,j$. Let $c'$ be a honest sanitization of a honest generated encryption of message $m$ under identity $i$: 
\begin{align*}
c' = \San(rk,\Enc(ek_i,m)) 
	= \sanFE.\San(pp^\sanFEpow,\sanFE.\Enc(pp^\sanFEpow,(m,i,F_{2}(ek_i,m))))
\end{align*}

Given the decryption key $dk_j = SK_j \from \sanFE.\Gen(msk,f_j)$. Then the correctness property of the $\sanFE$ scheme gives
\begin{align*}
	\Pr \left[ 
		\Dec(dk_j, c') = m
	\right] 
		= \Pr \left[ 
			\sanFE.\Dec(SK_j,c') = m
		\right] 
		\leq \negl(\kappa)
\end{align*}
\end{proof}

\begin{thm} \label{thm:ACE_sFE_no-read}
For any adversary $A$ that breaks the No-Read Rule of Construction~\ref{con:ACEfromFE}, there exists an adversary $B$ for the IND-CPA security of the sanitizable functional encryption scheme, such that the advantage of $A$ is
	\begin{align*}
	\adv{\ACE,A} \leq \adv{\sanFE,B}
	\end{align*}
\end{thm}

\begin{proof}
Assume that any adversary wins the IND-CPA security game of the sanitizable functional encryption (sFE) scheme with advantage at most $\epsilon$.
Assume for contradiction that there is an adversary $A$ that wins the ACE no-read game with advantage greater than $\epsilon$, then we can construct an adversary $B$ that wins the IND-CPA security game for the sFE scheme with advantage greater than $\epsilon$.

$B$ starts by generating $K \from \zo^\kappa$ for some pseudorandom function $F_1$. Then $B$ receives $pp^\sanFEpow$ from the challenger and forwards it as the ACE public parameter to the adversary $A$. 
Adversary $A$ then performs some oracle queries to $\mathcal{O}_G$ and $\mathcal{O}_E$ to which $B$ replies as follows: 
\begin{itemize}
\item $B$ receives $(j,\sen)$, then he sends $ek_j \from F_1(K,j)$ to $A$.
\item $B$ receives $(j,\rec)$, then he makes an oracle query $\mathcal{O}(f_j)$ to the challenger and gets back $SK_j$. $B$ sends $dk_j = SK_j$ to $A$.
\item $B$ receives $(j,\san)$, then he makes an oracle query $\mathcal{O}(f_{rk})$ to the challenger and gets back $SK_{rk}$. $B$ sends $rk = SK_{rk}$ to $A$.
\item $B$ receives $(i,m)$, then he computes $ek_i \from F_1(K,i)$ and sends to $A$ $$c \from \sanFE.\Enc(pp^\sanFEpow,(m,i,F_{2}(ek_i,m)))$$
\end{itemize}
After the oracle queries $B$ receives messages $m_0,m_1$ and identities $i_0,i_1$ from adversary $A$. Then $B$ computes $ek_{i_l} \from F_1(K,i_l)$ for $l \in \zo$ and sends $m_0^\sanFEpow$ and $m_1^\sanFEpow$ to the challenger, where $m_l^\sanFEpow = (m_l,i_l,F_{2}(ek_{i_l},m_l))$ for $l \in \zo$.
Then the sFE challenger sends a ciphertext $c'$, which $B$ forwards to $A$ as the ACE ciphertext. 
This is followed by a new round of oracle queries.

If the sFE challenger is in case $b=0$, then $c'$ is generated as an sFE encryption of message $m_0^\sanFEpow$, and we are in the case $b=0$ in the no-read game. Similar, if the sFE challenger is in case $b=1$, then we are in the case $b=1$ in the no-read game. Note that our adversary respects the rules of the IND-CPA game, since $f_{rk}(m_0^\sanFEpow)=f_{rk}(m_1^\sanFEpow)=1$ and $f_{j}(m_0^\sanFEpow)=f_{j}(m_1^\sanFEpow)$ for all $j$ such that $SK_j$ was queried. This follows directly from the payload privacy (the function outputs $\bot$) and sender anonymity ($m_0^\sanFEpow=m_1^\sanFEpow$) properties of the no-read rule.
Thus, we can conclude that if $A$ wins the no-read game with non-negligible probability, then $B$ wins the IND-CPA security game for the sFE scheme.
\end{proof}


\begin{thm} \label{thm:ACE_sFE_no_write}
For any adversary $A$ that breaks the No-Write Rule of Construction~\ref{con:ACEfromFE}, there exists an adversary $B$ for the PRF security, an adversary $C$ for the sanitizer property of the sFE scheme, and an adversary $D$ for the IND-CPA security of the sFE scheme, such that the advantage of $A$ is
	\begin{align*}
	\adv{\ACE,A} \leq 3 \cdot \adv{\PRF, B} + \adv{\sanFE,C} + \adv{\sanFE,D} + 2^{-\kappa}
	\end{align*}
\end{thm}

\begin{proof}
This theorem is proven by presenting a series of hybrid games.

\paragraph{Hybrid 0.} The no-write game for $b=1$

\paragraph{Hybrid 1.} As Hybrid 0, except that when the challenger receives a oracle request $(i,\sen)$ he saves the identity: $I_S = I_S \cup i$, and the encryption key $ek_i \from F_1(K,i)$. When the challenger receives the challenge $(c,i')$ he uses the sFE master decryption to get 
$$ (m^*,i^*,t^*) \from \sanFE.\MasterDec(msk^\sanFEpow,c)$$
If $i^* \notin I_S$, then the challenger generates $ek_{i^*}$ honestly.
Next, he checks that $t^* = F_{2}(ek_{i^*},m^*)$. If the check goes through he computes the challenge response as $c^* \from \sanFE.\San(pp^\sanFEpow,c)$, otherwise $c^* \from \San(rk,\Enc(ek_0,\bot))$.

\paragraph{Hybrid 2.} As Hybrid 1, except that the encryption keys are chosen uniformly at random: $ek_i \from_\$ \zo^\kappa$ for all $i$, (note that $ek_{i^*}$ is also chosen at random).

\paragraph{Hybrid 3.} As Hybrid 2, except that after receiving and master decrypting the challenge, the challenger check whether $i^* \in I_S$. If this is the case the challenger checks the MAC $t^*$ as above, otherwise he compute the response as  $c^* \from \San(rk,\Enc(ek_0,\bot))$.

\paragraph{Hybrid 4.} As Hybrid 3, except that if the checks $i^* \in I_S$ and $t^* = F_{2}(ek_{i^*},m^*)$ go through, then the challenger computes the response as
$$ c^* \from \sanFE.\San(pp^\sanFEpow,\sanFE.\Enc(pp^\sanFEpow, (m^*,i^*,t^*)))$$

\paragraph{Hybrid 5.} As Hybrid 4, except that the challenge response is computed as
$$ c^* = \San(rk,\Enc(ek_{i'},r))$$
where $r \from_\$ \M$ and $rk \from \Gen(msk,n+1,\san)$.

\paragraph{Hybrid 6.} As Hybrid 5, except that the encryption keys are generated honestly: $ek_i \from F_1(K,i)$ for all $i$. Observe, this is the no-write game for $b=0$. \\


Now we show that each sequential pair of the hybrids are indistinguishable.

\begin{claim}
Hybrid 0 and Hybrid 1 are identical.
\end{claim}

\begin{proof}
This follows directly from the definition of the sanitization and sanitizer key $rk$.
\end{proof}

\begin{claim} \label{claim:ACE_nowrite_PRF_1}
For any adversary $A$ that can distinguish Hybrid 1 and Hybrid 2, there exists an adversary $B$ for the security of PRF $F_1$ such that the advantage of $A$ is 
$ \adv{A} \leq \adv{\PRF,B}$.
\end{claim}

\begin{proof}
Assume that any adversary can break the PRF security with advantage $\epsilon$, and assume for contradiction that we can distinguish the hybrids with advantage greater than $\epsilon$. Then we can construct an adversary $B$ that breaks the PRF security with advantage greater than $\epsilon$.

$B$ starts by creating the public parameters honestly and sends it to the adversary. 
All the adversary oracle queries are answered as follows: whenever $B$ receives $(i,\sen)$ from the adversary, he sends $i$ to the PRF challenger, receives back $y_i$, set $ek_i := y_i$, and sends $ek_i$ to the adversary. 
When $B$ receives the challenge $(i,m)$ he ask the challenger for the encryption key (as before), and encrypts $m$. The rest of adversary's queries are answered honestly by using the algorithms of the construction.
When $B$ receives $(c,i')$ from the adversary, he master decrypts the ciphertext to get $(m^*,i^*,t^*)$. If $i^* \notin I_S$, then $B$ creates $ek_{i^*}$ by sending $i^*$ to the challenger.
$B$ concludes the game by forwarding the adversary's guess $b'$ to the challenger.

Observe that the if $y_i \from F_1(K,i)$ then we are in Hybrid 1, and if $y_i$ is uniform random, then we are in Hybrid 2. 
Thus, if adversary $A$ can distinguish between the hybrids, then $B$ can break the constraint PRF property. 
\end{proof}

\begin{claim} 
For any adversary $A$ that can distinguish Hybrid 2 and Hybrid 3,
there exists an adversary $B'$ for the security of PRF $F_2$ such that the advantage of $A$ is 
$ \adv{A} \leq \adv{\PRF,B'} + 2^{-\kappa}$.
\end{claim}

\begin{proof}
Assume that any adversary can break the PRF security with advantage $\epsilon - 2^{-\kappa}$, and assume for contradiction that we can distinguish the hybrids with advantage greater than $\epsilon$. Then we can construct an adversary $B'$ that breaks the PRF security with advantage greater than $\epsilon - 2^{-\kappa}$.

$B'$ starts by creating the public parameters and sending them to the adversary. The adversary's oracle queries are answered honestly by using the algorithms of the construction.
When $C$ receives the challenge $(c,i')$ he master decrypts the ciphertext to get $(m^*,i^*,t^*)$. Then he sends $m^*$ to the challenger and receives back $t'$. If $t' = t^*$ then $B'$ guess that the challenger is using the pseudorandom function $F_2$, otherwise $B'$ guess that the challenger is using a random function.

We evaluate now the advantage of $B'$ in the PRF game: Observe, if $t'$ is generated using $F_2$, then $B'$ outputs ``PRF'' with probability exactly $\epsilon$. In the case when $t'$ is generated using a random function, then it does not matter how $t^*$ was created, and the probability that $t' = t^*$ is $2^{-\kappa}$. Thus, the advantage of adversary $B'$ is greater than $\epsilon - 2^\kappa$.
\end{proof}

\begin{claim} \label{claim:ACE_nowrite_sFE_san}
For any adversary $A$ that can distinguish Hybrid 3 and Hybrid 4, there exists an adversary $C$ for the sanitizer property of the sFE scheme such that the advantage of $A$ is 
$ \adv{A} \leq \adv{\sanFE,C} $.
\end{claim}

\begin{proof}
Assume that any adversary wins the sanitizer game for the sFE scheme with advantage $\epsilon$, and assume for contradiction that we can distinguish the hybrids with advantage greater than $\epsilon$. Then we can construct an adversary $C$ that wins the sanitizer game with advantage greater than $\epsilon$.

$C$ starts by receiving the sFE system parameters from the challenger, and he forwards the public parameters as the ACE public parameters to the adversary. 
The adversary's oracle queries are answered honestly by using the algorithms of the construction, since $C$ receives the sFE master secret key from the challenger.
When $C$ receives the challenge $(c,i')$ he master decrypts the ciphertext to get $(m^*,i^*,t^*)$. Then he checks that $i^* \in I_S$ and $t^* = F_{2}(ek_{i^*},m^*)$. If the check goes through he sends $c$ to the challenger and receives back a sFE sanitized ciphertext $c'$. Thus, the challenge response is $c^* = c'$. 
$C$ concludes the game by forwarding the adversary's guess $b'$ to the challenger.

Observe, if $c' = \sanFE.\San(pp^\sanFEpow,c)$, then we are in Hybrid 3. On the other hand, we are in Hybrid 4 if $$c' = \sanFE.\San(pp^\sanFEpow,\sanFE.\Enc(pp^\sanFEpow,\sanFE.\MasterDec(msk^\sanFEpow,c)))$$

Thus, if adversary $A$ can distinguish between the hybrids, then $C$ can break the sFE sanitizer property. 
\end{proof}

\begin{claim} \label{claim:ACE_nowrite_sFE_IND-CPA}
For any adversary $A$ that can distinguish Hybrid 4 and Hybrid 5, there exists an adversary $D$ for the IND-CPA security of the sFE scheme such that the advantage of $A$ is 
$ \adv{A} \leq \adv{\sanFE,D} $.
\end{claim}

\begin{proof}
Assume that any adversary wins the IND-CPA game for the sFE scheme with advantage $\epsilon$, and assume for contradiction that we can distinguish the hybrids with advantage greater than $\epsilon$. Then we can construct an adversary $D$ that wins the IND-CPA game with advantage greater than $\epsilon$.

$D$ start by receiving the sFE public parameters from the challenger and forwards it to the challenger. The adversary's oracle queries are answered by sending secret key queries to the challenger, and otherwise using the algorithms of the construction (see the proof of Theorem~\ref{thm:ACE_sFE_no-read} for more details).
When $D$ receives the challenge $(c,i')$ he master decrypts the ciphertext to get $(m^*,i^*,t^*)$. Then he checks that $i^* \in I_S$ and $t^* = F_{2}(ek_{i^*},m^*)$. If the check goes through he set $m_0 = (m^*,i^*,t^*)$, otherwise he sets $m_0 = (\bot,0,\bot)$.
Then he creates $m_1 = (r,i',F_{2}(ek_{i'},r))$, sends $m_0$ and $m_1$ to the challenger, and receives back an sFE encryption $c'$. 
Next, $D$ creates the response $c^* = \sanFE.\San(pp^\sanFEpow,c')$.
$D$ concludes the game by forwarding the adversary's guess $b'$ to the challenger.

If $c'$ is an encryption of the message $m_0$, then we are in Hybrid 4, and if it is an encryption of $m_1$, then we are in Hybrid 5.
Thus, if adversary $A$ can distinguish between the hybrids, then $D$ can break the sFE IND-CPA security. 
\end{proof}

\begin{claim} \label{claim:ACE_nowrite_PRF_2}
For any adversary $A$ that can distinguish Hybrid 5 and Hybrid 6, there exists an adversary $B$ for the security of PRF $F_1$ such that the advantage of $A$ is 
$ \adv{A} \leq \adv{\PRF,B}$. \\
\end{claim}

The proof follow the same structure as the proof for Claim~\ref{claim:ACE_nowrite_PRF_1}. \\

From these claims we can conclude that for any adversary $A$ that can distinguish Hybrid 0 and Hybrid 6, there exists an adversary $B$ for the PRF security, an adversary $C$ for the sanitizer property of the sFE scheme, and an adversary $D$ for the IND-CPA security of the sFE scheme, such that the advantage of $A$ is
\begin{align*}
	\adv{\ACE,A} \leq 3 \cdot \adv{\PRF, B} + \adv{\sanFE,C} + \adv{\sanFE,D} + 2^{-\kappa}
\end{align*}
\end{proof}



\newpage
\bibliographystyle{alpha}
\bibliography{ACE}
\appendix

\clearpage

\section{Standard Building Blocks}\label{app:bb}

\subsection{Pseudorandom Function}
\begin{defi}[PRF]
We say $F : \zo^\kappa \times \zo^* \to \zok$ is a pseudorandom function if for all PPT $A$ 
$$\adv{A} = 2\cdot|\Pr[A^{\mathcal{O}_b(\cdot)}(1^\kappa)=b]-1/2|<\negl(\kappa)$$ with $\mathcal{O}_0$ a uniform random function and $\mathcal{O}_1=F_K$.  
\end{defi}

\subsection{Statistical Simulation-Sound Non-Interactive Zero-Knowledge Proofs} \label{app:nizk_sss}
The content of this subsection is taken almost verbatim from~\cite{DBLP:conf/focs/GargGH0SW13}.
Let $L$ be a language and $R$ a relation such that $x \in L$ if and only if there exists a witness $w$ such that $(x,w) \in R$.
A non-interactive proof system \cite{DBLP:conf/stoc/BlumFM88} for a relation $R$ is defined by the following PPT algorithms
\begin{description}
\item[Setup:] The $\Setup$ algorithm takes as input the security parameter $\kappa$ and outputs common reference string $crs$.

\item[Prove:] The $\Prove$ algorithm takes as input the common reference string $crs$, a statement $x$, and a witness $w$, and outputs a proof $\pi$.

\item[Verify:] The $\Verify$ algorithm takes as input the common reference string $crs$, a statement $x$, and a proof $\pi$. It outputs $1$ if it accepts the proof, and $0$ otherwise.
\end{description}

The non-interactive proof system must be complete, meaning that if $R(x,w)=1$  and $crs \from \Setup(1^\kappa)$ then $$\Verify(crs,x,\Prove(crs,x,w)) = 1$$  
Furthermore, the proof system must be statistical sound, meaning that no (unbounded) adversary can convince a honest verifier of a false statement.
Moreover, we define the following additional properties of a non-interactive proof system.

\begin{defi}[Computational Zero-Knowledge]
A non-interactive proof $\NIZK = (\Setup, \Prove, \Verify)$ is computational zero-knowledge if there exists a polynomial time simulator $\Sim = (\Sim_1,\Sim_2)$ such that for all non-uniform polynomial time adversaries $A$ we have for all $x \in L$ that
	\begin{center}
		$\Pr \left[ crs \from \Setup(1^\kappa); \pi \from \Prove(crs,x,w): A(crs,x,\pi) = 1 \right]$\\
		$\approx$ \\
		$\Pr \left[ (crs,\tau) \from \Sim_1(1^\kappa,x); \pi \from \Sim_2(crs,\tau,x): A(crs,x,\pi) = 1 \right]$		
	\end{center}
where $crs$ is the common reference string, $x$ is the statement, $w$ is the witness, $\pi$ is the proof, and $\tau$ is the trapdoor.
\end{defi}

Thus, the definition states that the proof do not reveal any information about the witness to any bounded adversary.
In the definition this is formalized by the existence of two simulators, where $\Sim_1$ returns a simulated common reference string together with a trapdoor that enables $\Sim_2$ to simulate proofs without access to the witness.

\begin{defi}[Statistical Simulation-Soundness]
A non-interactive proof $\NIZK = (\Setup, \Prove, \Verify)$ is statistical simulation-sound (SSS) if for all statements $x$ and all (unbounded) adversaries $A$ we have that
	\begin{align*}
	\Pr \left[ \begin{array}{c}
		(crs,\tau) \from \Sim_1(1^\kappa,x); \pi \from \Sim_2(crs,\tau,x):  \\
			\exists (x',\pi'): x' \neq x: \Verify(crs,x',\pi') = 1: x' \notin L
	\end{array} \right] \leq \sssprob 
	\end{align*}
	where $\sssprob = \negl(\kappa)$ is negligible in the security parameter.
\end{defi}

Thus, the definition states that it is not possible to convince a honest verifier of a false statement even if the adversary is given a simulated proof.

\begin{remark}
If a proof system is statistical simulation-sound then it is also statistical sound. Thus, we can upper bound the negligible probability of statistical soundness by the negligible probability of the statistical simulation-soundness. 
\end{remark}
\subsection{Indistinguishability Obfuscation}
\newcommand{\Obf}{iO}
\newcommand{\Circs}{\mathcal{C}}
\newcommand{\obf}[1]{\bar{#1}}
\newcommand{\poly}{\mathsf{poly}}
\newcommand{\Adv}{\mathcal{A}}

We use an \emph{indistinguishability obfuscator} like the one proposed in~\cite{DBLP:conf/focs/GargGH0SW13} such that $\obf{C}\from \Obf(C)$ which takes any polynomial size circuit $C$ and outputs an obfuscated version $\obf{C}$ that satisfies the following property.
\begin{defi}[Indistinguishability Obfuscation] \label{def:io} We say $\Obf$ is an \emph{indistinguishability obfuscator} for a circuit class $\Circs$ if for all $C_0,C_1\in\Circs$ such that $\forall x:C_0(x)=C_1(x)$ and $|C_0|=|C_1|$ it holds that:
\begin{enumerate}
\item $\forall C\in \Circs, \forall x\in\zon,  \Obf(C)(x)=C(x)$;
\item $|\Obf(C)|=\poly(\lambda|C|)$  
\item  for all PPT $\Adv$:
$$
\adv{\Adv} = 2\cdot \left | \Pr[\Adv(\Obf(C_0)) = 1]-\Pr[\Adv(\Obf(C_1)) = 1]\right| < \negl(\lambda)
$$
\end{enumerate}
\end{defi}



\section{Equivalent Definition of the No-Write Rule}
In this section we provide an alternative definition of the No-Write rule and we prove that this is equivalent to the No-Write rule of Definition~\ref{def:ACEnowrite} (Section~\ref{sec:definition}). The alternative definition is used in the proof of Theorem~\ref{thm:no-write}.
In Definition~\ref{def:ACEnowrite} the challenger chooses, encrypts and sanitizes a random message if $b=0$. In the following definition, we will let the adversary choose the message to be encrypted and sanitized in case $b=0$. Thus, we replace the randomly chosen message with an adversarial chosen message.

\begin{defi}[Alternative No-Write Rule] \label{def:alternativeACEnowrite}
Consider the following game between a challenger $C$ and a stateful adversary $A$:
\begin{center}
\begin{small}
    \begin{tabular}{| p{6.5cm} | p{5cm} |}
    \hline
	\multicolumn{2}{|c|}{\textbf{\emph{No-Write Rule}}} \\
	\hline
	\multicolumn{1}{|c|}{\textbf{\emph{Game Definition}}} & \multicolumn{1}{|c|}{\textbf{\emph{Oracle Definition}}}  \\
	\hline
	\

	1. $(pp,msk)\from \Setup(1^\kappa,P)$;

	2. $(c,i',m) \from A^{\mathcal{O}_{E}(\cdot),\mathcal{O}_{S}(\cdot)}(pp)$;

	3. $ek_{i'} \from \Gen(msk,i',\sen)$;

	4. $rk\from \Gen(msk,n+1,\san)$;

	5. $b\from \zo$, 
	\begin{itemize}
		\item If $b=0$, $c'\from \San(rk,\Enc(ek_{i'},m))$; 
		\item If $b=1$, $c'\from \San(rk,c)$;
	\end{itemize}
	6. $b' \from A^{\mathcal{O}_{E}(\cdot),\mathcal{O}_{R}(\cdot)}(c')$;

	&
	\
	
	$\mathcal{O}_S(j,t)$: \newline 1. Output $k\from \Gen(msk,j,t)$;

	\ 

	$\mathcal{O}_R(j,t)$: \newline 1. Output $k\from \Gen(msk,j,t)$;

	\

	$\mathcal{O}_E(i,m)$: 

	1. $ek_i\from  \Gen(msk,i,\sen)$;

	2. $c\from \Enc( ek_i , m)$;
	
	3. Output $c'\from \San(rk,c)$;

	\\
	\hline
	\end{tabular}
\end{small}
\end{center}
Let $Q_S$ (resp. $Q$) be the set of all queries $q=(j,t)$ that $A$ issues to $\mathcal{O}_S$ (resp. both $\mathcal{O}_S$ and $\mathcal{O}_R$). Let $I_S$ be the set of all $i\in[n]$ such that $(i,\sen)\in Q_S$ and let $J$ be the set of all $j\in [n]$ such that $(j,\rec)\in Q$. Then we say that $A$ \emph{wins the No-Write game} if $b'=b$ and all of the following hold:
\begin{enumerate}
	\item $(n+1,\san)\not\in Q$;
	\item $i'\in I_S \cup \{0\}$;
	\item $\forall i\in I_S,j\in J$,  $P(i,j)=0$;
\end{enumerate}
We say an ACE scheme satisfies the \emph{No-Write rule} if for all PPT $A$
$$
\adv{A} = 2\cdot \left|\Pr[A \mbox{ wins the No-Write game}]-\frac{1}{2}\right|\leq \negl(\kappa)
$$ 
\end{defi}

\newcommand{\ACEscheme}{\mathsf{ACE}}
\begin{lem} \label{lem:nowrite_equiv}
Let $\ACEscheme$ be a correct ACE scheme that satisfies the \emph{No-Read Rule}. Then $\ACEscheme$ satisfies the alternative \emph{No-Write Rule} from Def.~\ref{def:alternativeACEnowrite} iff it satisfies the \emph{No-Write Rule} from Def.~\ref{def:ACEnowrite}.
\end{lem}

\setcounter{claimcounter}{0}
\begin{proof}
The lemma is proven by splitting the bi-implication in two cases:

\ \newline
\noindent\textbf{Def.~\ref{def:alternativeACEnowrite} implies Def.~\ref{def:ACEnowrite}.}
Assume that $\ACEscheme$ satisfies the alternative no-write rule, and assume for contradiction that there is an adversary $A$ that wins the ACE no-write game (from Def.~\ref{def:ACEnowrite}), then we can construct an adversary $B$ that wins the alternative ACE no-write game.

Adversary $B$ receives the public parameters from the challenger (from Def.~\ref{def:alternativeACEnowrite}) and forwards them to the adversary $A$. Then $A$ performs some oracle queries, which adversary $B$ forwards to the challenger. After the oracle queries $B$ receives the challenge $(c,i')$ from adversary $A$. Next, $B$ picks a message $r \from \M$ and sends $(c,i',r)$ to the challenger. 
The challenger respond with a sanitized ciphertext $c'$, where $c' = \San(rk,\Enc(ek_{i'},r))$ if $b=0$, and $c' = \San(rk,c)$ if $b=1$.
Adversary $B$ forward $c'$ to $A$. This is followed by a new round of oracle queries. Note that $c'$ corresponds to the challenge of the no-write game of Def.~\ref{def:ACEnowrite}.
Thus, we can conclude that if $A$ wins the ACE no-write game (from Def.~\ref{def:ACEnowrite}), then $B$ wins the alternative ACE no-write game (from Def.~\ref{def:alternativeACEnowrite}).

\ \newline
\noindent\textbf{No-Read Rule and Def.~\ref{def:ACEnowrite} implies Def.~\ref{def:alternativeACEnowrite}.}
This is proven by presenting a series of hybrid games.

\paragraph{Hybrid 0.} The alternative no-write game (Def.~\ref{def:alternativeACEnowrite}) for $b=0$.

\paragraph{Hybrid 1.} As Hybrid 0, except the challenger ignores the message $m$ send by the adversary and draws its own message $r$ to encrypt.

\paragraph{Hybrid 2.} As Hybrid 1, except $b=1$.

\paragraph{Hybrid 3.} The alternative no-write game (Def.~\ref{def:alternativeACEnowrite}) for $b=1$.

\begin{claim}
Assume that the ACE scheme $\ACEscheme$ satisfies the no-read rule, then Hybrid 0 and Hybrid 1 are indistinguishable.
\end{claim}

\begin{proof}
Assume for contradiction that there is an adversary $A$ that can distinguish between the hybrids, then we can construct an adversary $B$ that wins the no-read game from Def.~\ref{def:ACEnoread}.

Adversary $B$ forwards the public parameter from the challenger (of the no-read game) to adversary $A$, and forwards the oracle queries made by $A$ to the challenger. Note that $A$ plays the alternative no-write game, where the allowed oracle queries are a strict subset of the allowed oracle queries in the no-read game. 
Next, $A$ sends the challenge $(c,i',m)$ to $B$, who draws a random message $r \from \M$ and sends $(m,r,i',i')$ to the challenger. 
The challenger draws a random bit $b^*$ and responds with $c^* \from \Enc(ek_{i'},m)$ if $b^*=0$, and $c^* \from \Enc(ek_{i'},r)$ if $b^*=1$. 
Adversary $B$ queries the challenger for the sanitizer key $rk$, sanitizes the challenge ciphertext $c' = \San(rk,c^*)$ and sends $c'$ to adversary $A$, who performs a new set of oracle queries. 
Note, if $b^*=0$ then we are in Hybrid 0 and if $b^*=1$ then we are in Hybrid 1, which means that if adversary $A$ can distinguish between the hybrids, then adversary $B$ breaks the no-read rule. Thus, Hybrid 0 and Hybrid 1 are indistinguishable.
\end{proof}

\begin{claim}
Assume that the ACE scheme $\ACEscheme$ satisfies the no-write rule from Def.~\ref{def:ACEnowrite}, then Hybrid 1 and Hybrid 2 are indistinguishable.
\end{claim}

\begin{proof}
Assume for contradiction that there is an adversary $A$ that can distinguish between the hybrids, then we can construct an adversary $B$ that wins the no-write game (from Def.~\ref{def:ACEnowrite}).

Adversary $B$ forwards the public parameter from the challenger (from Def.~\ref{def:ACEnowrite}) to adversary $A$, and forwards the oracle queries made by $A$ to the challenger. 
Next, $B$ receives the challenge $(c,i',m)$ from adversary $A$ and forward $(c,i')$ to the challenger.
The challenger draws a random bit $b^*$ and responds with $c' = \San(rk,\Enc(ek_{i'},r))$ for random $r \from \M$ if $b^*=0$, and $c' = \San(rk,c)$ if $b^*=1$. 
Adversary $B$ forwards $c'$ to $A$, which is followed by a new round of oracle queries. 
Note, if $b^*=0$ then we are in hybrid 1, and if $b^*=1$ then we are in hybrid 2, which means that if adversary $A$ can distinguish between the hybrids, then adversary $B$ breaks the no-read rule. Thus, Hybrid 1 and Hybrid 2 are indistinguishable.
\end{proof}

\begin{claim}
Hybrid 2 and Hybrid 3 are identical.
\end{claim}

\begin{proof}
In both hybrids the adversary sends the challenge $(c,i',m)$ and the challenger will in both cases respond with $c' = \San(rk,c)$.
\end{proof}

Thus, we can conclude that the alternative definition of the no-write rule presented in Def.~\ref{def:alternativeACEnowrite} is equivalent to the no-write rule presented in Def.~\ref{def:ACEnowrite} in Section~\ref{sec:definition}.
\end{proof}


\section{Linear ACE}

\subsection{ACE for a Single Identity from Pailler}\label{app:onepailler}

Our second instantiation of \oACE is based on Pailler's cryptosystem~\cite{DBLP:conf/eurocrypt/Paillier99,DBLP:conf/pkc/DamgardJ01}. The scheme uses the same high level idea as the DDH-based instantiation.

\begin{construction} \label{con:onePailler} Let $\mathsf{PACE} = (\Setup, \Gen, \Enc, \San, \Dec)$ be a \oACE scheme defined by the following algorithms:
\end{construction}
\begin{description}
\item[Setup and Key Generation:] The public parameters $pp$ contain the modulus $N$ and the master secret key $msk$ is the factorization of $N$. The encryption key for the only identity in the system is $ek=\alpha$ for a random $\alpha \from \mathbb{Z}_N$ and the sanitizer key is $rk=-\alpha$. Finally the decryption key $dk$ is the master secret key $msk$.
Furthermore, the message space is $\M=\mathbb{Z}_N$ and the ciphertext spaces are $\C = \mathbb{Z}^*_{N^2}\times\mathbb{Z}^*_{N^2}$ and $\C' = \mathbb{Z}^*_{N^2}$.

\item[Encryption:] To encrypt a message $m\in \M$ with identity $1$ first sample $$(r_0,r_1)\from\mathbb{Z}^*_N \times \mathbb{Z}^*_N$$ and then output:

$$
(c_0,c_1)=( (1+ekN) r_0^N, (1+mN) r_1^N)
$$
(and encryptions for the identity $0$ are random $(c_0,c_1)\from \C$.)

\item[Sanitizer:] On input $(c_0,c_1)\in \C$ and the randomization key $rk$ sample $(\beta,s) \in\mathbb{Z}_N \times  \mathbb{Z}^*_N $ and then output 
$$
c'= c_1\cdot (c_0 \cdot (1+rkN) )^\beta \cdot s^N
$$

\item[Decryption:] On input $c'\in \C'$ and the decryption key run the decryption of Pailler cryptosystem to get $m' \in \M$ from $c'$. 
\end{description}

\begin{lem} Construction~\ref{con:onePailler} is a \emph{correct} \oACE scheme satisfying the \emph{No-Read} and the \emph{No-Write Rule} assuming the Pailler's assumption holds.
\end{lem}

\begin{proof} \emph{Correctness:} Correctness follows from inspection: thanks to the homomorphic properties of Pailler $m' = m + (\alpha-\alpha)\beta=m$. \\

\noindent \emph{No-Read Rule:} Exactly as in Lemma~\ref{lem:elgamal} we only need to prove that honest encryptions are indistinguishable from a random element in $(\mathbb{Z}_{N^2}\times\mathbb{Z}_{N^2})$, which follows in a straightfoward way from the Pailler assumption: both $c_0$ and $c_1$ are fresh Pailler encryptions using independent random values $r_0,r_1$, and the assumption says that $r^N$ with $r\from \mathbb{Z}^*_N$ is indistinguishable from a random element in $\mathbb{Z}_{N^2}$. \\

\noindent \emph{No-Write Rule:} As in Lemma~\ref{lem:elgamal} there are only two cases, depending on which keys the adversary asks for \emph{before} producing the challenge ciphertext $c$ and identity $i'$: 1) the adversary asks for $ek$ before issuing his challenge $(c,i')$ with $i'\in\zo$ (and receives no more keys during the distinguishing phase) and 2) the adversary asks for $dk$ before issuing his challenge $(c,0)$ and then asks for $ek$ during the distinguishing phase. Case 1) follows directly from the security of Pailler cryptosystem: without access to the decryption key the output of the sanitizer is indistinguishable from a random ciphertext thanks to the choice of the random $s$. In case 2) instead the adversary has the decryption key, and we therefore need to argue that the distribution of $\San(rk,(c_0,c_1))$ is independent of its input unconditionally. Given any $(c_0,c_1)\in  \C$ we can write:
$$
(c_0,c_1)=((1+\delta_0 N)t_0^N, (1+\delta_1 N)t_1^N)
$$
Then the output $c'\from \San(rk,(c_0,c_1))$ is
\begin{align*}
	c' &= c_1\cdot (c_0 \cdot (1+rkN) )^\beta \cdot s^N  \\
  	   &= (1+(\delta_1+\beta(\delta_0-\alpha)) N)(t_1t_0^\beta s)^N 
\end{align*}
Which is distributed exactly as a uniformly random ciphertext $(1+\gamma N)u^N$ with $(\gamma,u)\in  \mathbb{Z}_{N}\times\mathbb{Z}^*_{N}$ since for all $(\gamma,u)$ there exists $(\beta,s)$ such that 
$$
\gamma = \delta_1 + \beta (\delta_0-\alpha) \mbox{ and } u=t_1t_0^\beta s
$$
Which is guaranteed unless $\delta_0=\alpha$ which happens only with negligible probability thanks to the principle of deferred decisions.
\end{proof}


\subsection{The Repetition Scheme - Proofs}\label{app:repetitionproofs}

\begin{lem}
Assume that $\oneACE$ is a correct \oACE scheme.
Then the ACE scheme $\ACE$ from Construction~\ref{con:ACE} enjoys correctness.
\end{lem}

\begin{proof}
Given any $i,j \in [n]$ such that $P(i,j) = 1$.
Let $c = \left( c_1^\oneACEpow, \dots, c_n^\oneACEpow  \right)$ be an honest encryption of $m$ under encryption key $ek_i$, and let $c' := ({c'_1}^\oneACEpow,\dots,{c'_n}^\oneACEpow)$ be a honest sanitized version of $c$.
Thus, the ciphertext $c_j^\oneACEpow$ is a honest created \oACE encryption (since $ek_j^\oneACEpow \in ek_i$), and ${c'_j}^\oneACEpow$ is a sanitized version of $c_j^\oneACEpow$. 
Observe, given the decryption key $dk_j = dk_j^\oneACEpow$, the decryption algorithm $\ACE.\Dec$ decrypts the $j$'th $\oneACE$ ciphertext.
This means that
\begin{align*}
	\Pr \left[ \ACE.\Dec(dk_j,c') \neq m \right] &= \Pr \left[ \oneACE.\Dec(dk_j^\oneACEpow,{c'_j}^\oneACEpow) \neq m \right]
		< \negl(\secparam)
\end{align*}

The last inequality comes from the correctness of the \oACE scheme.
Thus, we can conclude that $\ACE$ enjoys correctness.
\end{proof}

\begin{thm} \label{thm:no-read}
Assume that $\oneACE$ is a \oACE scheme that satisfies the \emph{No-Read Rule}, and let $\ACE$ be the ACE scheme from Construction~\ref{con:ACE} using $\oneACE$ as the underlying \oACE scheme. Then $\ACE$ satisfies the \emph{No-Read Rule}.
\end{thm}

\begin{proof}
Let $Q_G$ be the set of all queries $q = (j,t)$ that the adversary issues to the oracle $\mathcal{O}_G$, and let $J$ be the set of all $j\in[n]$ such that $q = (j,\rec) \in Q_G$. 

The no-read rule is shown by presenting a series of hybrid games such that no adversary can distinguish between two successive hybrid games with non-negligible probability.
In the hybrid games we  replace the encryption algorithm with the algorithm: $\Enc_k^*(ek_{i_0},ek_{i_1},m_0,m_1)$, which starts by encrypting the two messages $m_0$ and $m_1$ under the keys $ek_{i_0}$ and $ek_{i_1}$ to get 
\begin{align*}
c^0 &= (c_1^0,\dots,c_n^0) \gets \Enc(ek_{i_0},m_0) \\
c^1 &= (c_1^1,\dots,c_n^1) \gets \Enc(ek_{i_1},m_1)
\end{align*}
Next, the algorithm outputs the following ciphertext $(c_1^1,\dots,c_k^1,c_{k+1}^0,\dots,c_n^0)$, where the first $k$ positions are the first $k$ \oACE ciphertexts from $c^1$, and the last $n-k$ positions are the last $n-k$ \oACE ciphertexts from $c^0$.

\paragraph{Game 0.} The no-read game with $b=0$;
\paragraph{Hybrid $k$ for $k=0,\dots,n$.} Like Game 0 but the encryption algorithm is replaced by the $\Enc_k^*$ algorithm;
\paragraph{Game 1.} The no-read game with $b=1$;

\ \newline
\noindent\textbf{Payload Privacy.} 
The conditions for the payload privacy property states that for all $j \in J$ it holds that $P(i_0,j) = P(i_1,j) = 0$.
This means that for all $j \in J$ we have that $ek_j^\oneACEpow \notin ek_{i_s}$ for $s \in \{0,1 \}$.

\paragraph{Game 0 $\approx$ Hybrid 0.} 
In  Game 0 the ciphertext is generated by running $c \gets \Enc(ek_{i_0},m_0)$, while in Hybrid 0 the ciphertext is generated by running $$c' = (c_1^0,\dots,c_n^0) \gets \Enc_0^*(ek_{i_0},ek_{i_1},m_0,m_1)$$
From the description of the algorithm $\Enc_0^*$ we can conclude that $c'$ is an encryption of $m_0$ generated by the $\Enc$ algorithm using key $ek_{i_0}$.
Thus, the two ciphertexts are identically distributed, which means that Game 0 and Hybrid 0 are identical and therefore indistinguishable to any adversary.

\paragraph{Hybrid $k-1$ $\approx$ Hybrid $k$.} 
The only difference between the two hybrids is the following: in Hybrid $k-1$ the $k$'th \oACE ciphertext is the $k$'th ciphertext from $\Enc(ek_{i_0},m_0)$, while in Hybrid $k$ the $k$'th \oACE ciphertext is the $k$'th ciphertext from $\Enc(ek_{i_1},m_1)$.

Observe that in $\Enc(ek_{i_s},m_s)$ for $s \in \{0,1 \}$ the $k$'th ciphertext is a \oACE encryption of message $m_s$ under key $ek_k^\oneACEpow$ if $ek_k^\oneACEpow \in ek_{i_s}$, otherwise the ciphertext is taken uniformly random from the \oACE ciphertext space $\C_k^\oneACEpow$.
Thus, we look at four cases
\begin{enumerate}
\item $c_k^0 \gets \C_k^\oneACEpow$ and $c_k^1 \gets \C_k^\oneACEpow$
\item $c_k^0 \gets \C_k^\oneACEpow$ and $c_k^1 \gets \oneACE.\Enc(ek_k^\oneACEpow,m_1)$
\item $c_k^0 \gets \oneACE.\Enc(ek_k^\oneACEpow,m_0)$ and $c_k^1 \gets \C_k^\oneACEpow$
\item $c_k^0 \gets \oneACE.\Enc(ek_k^\oneACEpow,m_0)$ and $c_k^1 \gets \oneACE.\Enc(ek_k^\oneACEpow,m_1)$
\end{enumerate}

Notice, from the condition of the the payload privacy we have that the adversary is only allowed to ask for the decryption key $dk_k$ if $ek_k^\oneACEpow \notin ek_{i_s}$ for $s \in \zo$. 
Thus, if the adversary gets $dk_k$, then we are in Case 1.

In Case 1 the two ciphertexts are clearly indistinguishable, since they are chosen uniformly random from the same ciphertext space. 

In Case 2 and 3, the adversary does not have the decryption key $dk_k$. If the adversary $A$ is able to distinguish between Hybrid $k-1$ and Hybrid $k$, then we can construct a new adversary $B$ that breaks the no-read rule of the \oACE scheme. 
Intuitively this is done as follow: when $B$ receives the challenge from adversary $A$ he forwards the messages together with identities 1 and 0, and receives back a \oACE ciphertext $c'$ from the \oACE challenger. Then $B$ creates the challenge ciphertext as in Hybrid $k$, replaces the $k$'th \oACE ciphertext by $c'$, and sends the ciphertext to the adversary. 
$B$ will answer adversary $A$'s queries as follows
\begin{itemize}
\item $A$ queries $(i,\sen)$ where $P(i,k) = 1$, then $B$ queries the challenger for encryption key $ek_k^\oneACEpow$, and construct the rest of the encryption key $ek_i$ honestly.
\item $A$ queries $(n+1,\san)$, then $B$ queries the challenger for the sanitizer key $rk_k^\oneACEpow$, and construct the rest of the sanitizer key $rk$ honestly.
\item $A$ queries $(i,m)$ where $P(i,k) = b$ for $b\in\zo$, then $B$ sends the query $(b,m)$ to the challenger, receives back a ciphertext $c_k^\oneACEpow$, and construct the rest of the response honestly. 
\item For all other queries $B$ will answer using the algorithms of the scheme.
\end{itemize}
Notice, $A$ will never query $(k,\rec)$ since then we would not be in case 2 or 3. Thus, $B$ never has to query the challenger for $dk_k^\oneACEpow$.
If $A$ wins the no-read game with non-negligible probability, then so does adversary $B$. This contradicts our assumption that the \oACE scheme satisfies the no-read rule. Thus, adversary $A$ cannot distinguish between the two hybrids.



In Case 4, the adversary does not have the decryption key $dk_k$. To argue that the two ciphertexts are indistinguishable, we use an intermediate step. Following the same arguments as above, we argue that $c_k^0$ is indistinguishable from a random ciphertext $c^*$, and then we argue that $c_k^1$ is indistinguishable from $c^*$.
Thus, the two ciphertexts are indistinguishable.

\paragraph{Hybrid $n$ $\approx$ Game 1.} 
In  Game 1 the ciphertext is generated by running $c \gets \Enc(ek_{i_1},m_1)$, while in Hybrid $n$ the ciphertext is generated by running $$c' = (c_1^1,\dots,c_n^1) \gets \Enc_n^*(ek_{i_0},ek_{i_1},m_0,m_1)$$
From the description of the algorithm $\Enc_n^*$ we can conclude that $c'$ is actually an encryption of $m_1$ generated by the $\Enc$ algorithm using key $ek_{i_1}$.
Thus, the two ciphertexts are identically distributed, which means that Game 1 and Hybrid $n$ are identical and therefore indistinguishable to any adversary.

\ \newline
\noindent\textbf{Sender Anonymity.} 
The conditions for the sender anonymity property states that for all $j \in J$ it holds that $P(i_0,j) = P(i_1,j)$ and $m_0 = m_1$.
This means that for all $j \in J$ we have that $ek_j^\oneACEpow$ is either in both encryption keys $ek_{i_s}$ for $s \in \{0,1 \}$, or not in any of the two.

\paragraph{Game 0 $\approx$ Hybrid 0.} The same arguments as in the payload privacy case.

\paragraph{Hybrid $k-1$ $\approx$ Hybrid $k$.} 
The only difference between the two hybrids is the following: in Hybrid $k-1$ the $k$'th \oACE ciphertext is the $k$'th ciphertext from $\Enc(ek_{i_0},m_0)$, while in Hybrid $k$ the $k$'th \oACE ciphertext is the $k$'th ciphertext from $\Enc(ek_{i_1},m_1)$.
Again, we look at four cases depending on whether $ek_k^\oneACEpow \in ek_{i_s}$ for some $s \in \zo$
\begin{enumerate}
\item $c_k^0 \gets \C_k^\oneACEpow$ and $c_k^1 \gets \C_k^\oneACEpow$
\item $c_k^0 \gets \C_k^\oneACEpow$ and $c_k^1 \gets \oneACE.\Enc(ek_k^\oneACEpow,m_1)$
\item $c_k^0 \gets \oneACE.\Enc(ek_k^\oneACEpow,m_0)$ and $c_k^1 \gets \C_k^\oneACEpow$
\item $c_k^0 \gets \oneACE.\Enc(ek_k^\oneACEpow,m_0)$ and $c_k^1 \gets \oneACE.\Enc(ek_k^\oneACEpow,m_1)$
\end{enumerate}

Notice, from the condition of the sender anonymity property, we have that the adversary can ask for the decryption key $dk_k$, iff $P(i_0,k) = P(i_1,k)$ for all $k\in J$. Thus, if the adversary queries $dk_k$ then we are in Case 1 or 4. 

In Case 1, the two ciphertexts are indistinguishable, since they are chosen uniformly random from the same ciphertext space. 

In Case 2 and 3, the adversary does not have the decryption key $dk_k$. Thus, if the adversary can distinguish between the two hybrids, then he breaks the no-read rule of the \oACE scheme (see Case 2 and 3 for payload privacy). 

In Case 4, we have $P(i_0,k)=P(i_1,k)=1$, and the adversary is allowed to get the corresponding decryption key $dk_k$. If the adversary queries $dk_k$, then the condition states that $m_0 = m_1$. Thus, the two \oACE ciphertexts are encryptions of the same message under the same \oACE encryption key. Thus, the adversary cannot distinguish between the hybrids.
If the adversary does not ask for the decryption key $dk_k$, then we are in the same situation as in Case 4 for payload privacy.

\paragraph{Hybrid $n$ $\approx$ Game 1.} The same arguments as in the payload privacy case.

\end{proof}

\begin{thm} \label{thm:no-write}
Assume $\oneACE$ is a \oACE scheme that satisfies the \emph{No-Write Rule}, and let $\ACE$ be the ACE scheme from Construction~\ref{con:ACE} using $\oneACE$ as the underlying \oACE scheme. Then $\ACE$ satisfies the \emph{No-Write Rule}. 
\end{thm}

\begin{proof}
This theorem is shown by presenting a series of hybrid games such that no adversary can distinguish between two successive hybrids with non-negligible probability.

In the hybrid games we  replace the encryption/sanitization algorithm with a special challenge ciphertext generation algorithm: $\CCG_k$, which takes the encryption key $ek_{i'}$, the sanitizer key $rk$, a random element $r$ from the message space, and the ciphertext $c$ generated by the adversary. The algorithm then encrypts and sanitizes message $r$ to get 
$$(c_1^0,\dots,c_n^0) \from \ACE.\San(rk,\ACE.\Enc(ek_{i'},r))$$ 
and sanitizes the ciphertext $c$ to get 
$$(c_1^1,\dots,c_n^1) \from \ACE.\San(rk,c)$$ 
Finally, the algorithm outputs the following ciphertext $(c_1^1,\dots,c_k^1,c_{k+1}^0,\dots,c_n^0)$, where the first $k$ positions are sanitized encryptions of a random message, and the last $n-k$ are sanitizations of the adversary's ciphertext.

\paragraph{Game 0.} The no-write game with $b=0$
\paragraph{Hybrid $k$ for $k = 0,\dots,n$.} Like Game 0 but replace the encryption algorithm with the special challenge ciphertext generation algorithm $\CCG_k$.
\paragraph{Game 1.} The no-write game with $b=1$

\paragraph{Game 0 $\approx$ Hybrid 0.} 
In Hybrid 0 the ciphertext $c'$ is generated by running $\CCG_0$. From the description of this algorithm we can conclude that this is a ciphertext generated by $\ACE.\San(rk,\ACE.\Enc(ek_{i'},r))$. Thus, Game 0 and Hybrid 0 both generates the ciphertext $c'$ as a sanitized ACE encryption of a random message. Thus, Game 0 and Hybrid 0 are identical and therefore indistinguishable to any adversary.

\paragraph{Game $k-1$ $\approx$ Hybrid $k$.} The difference between the two hybrids is the following: in Hybrid $k-1$ the $k$'th \oACE ciphertext is the $k$'th ciphertext from $\ACE.\San(rk,\ACE.\Enc(ek_{i'},r))$, while in Hybrid $k$ the $k$'th \oACE ciphertext is the $k$'th ciphertext from $\ACE.\San(rk,c)$. 
This give rise to the following two cases depending on whether $ek_k^\oneACEpow \in ek_{i'}$
\begin{enumerate}
\item $c_k^0 \from \oneACE.\San(rk_k^\oneACEpow,\oneACE.\Enc(ek_{k}^\oneACEpow,r))$ \newline and $c_k^1 \from \oneACE.\San(rk_k^\oneACEpow,c_k^\oneACEpow)$
\item $c_k^0 \from {\C'_k}^\oneACEpow$ and $c_k^1 \from \oneACE.\San(rk_k^\oneACEpow,c_k^\oneACEpow)$
\end{enumerate}

Notice, from the condition stated by the no-write rule we have that for all $i \in I_S, j \in J$ it holds that $P(i,j) = 0$. This means that for all the encryption keys the adversary gets before the challenge and for all the decryption keys he gets during the game it must hold that the decryption keys cannot be used to decrypt anything encrypted using the encryption keys.
Specially, this means that if the adversary gets $dk_k$ then $ek_k^\oneACEpow \notin ek_{i'}$ (i.e.~Case 2).

In Case 1, the adversary does not get the decryption key $dk_k$. However, he has queried the encryption key $ek_{i'}$, which contains $ek_k^\oneACEpow$.
If the adversary $A$ is able to distinguish between Hybrid $k-1$ and Hybrid $k$, then he is able to distinguish between the case where the $k$'th \oACE ciphertext is created as a sanitized encryption of message $r$ or the sanitization of his ciphertext. This means that we can construct an adversary $B$ that wins the alternative No-Write game (Def.~\ref{def:alternativeACEnowrite}, equivalent to Def.~\ref{def:ACEnowrite} by Lemma~\ref{lem:nowrite_equiv}) for the \oACE scheme with non-negligible probability. (Note that here it is necessary to use the alternative No-Write game since the reduction must encrypt the same random value $r$ in all positions.) Intuitively, this is done by letting $B$ create the ciphertext $c'$ as in Hybrid $k$, forward the challenge $(c,i',r)$, and replacing the $k$'th \oACE ciphertext by the one he receives from the challenger. Finally, he sends the new ciphertext to the adversary $A$. 

If $A$ wins the game with non-negligible probability, then so does the adversary $B$. This gives us a contradiction with the assumption that the \oACE scheme satisfies the no-write rule. Thus, we can conclude that $A$ cannot distinguish between the two hybrids.
Furthermore, notice that the adversary $A$ can query encryption keys $ek_i$, where $ek_k^\oneACEpow \in ek_i$. In this case the adversary $B$  queries his challenger for the encryption key $ek_k^\oneACEpow$. This  never conflicts the no-write game for the \oACE scheme, since the adversary $A$ never queries the decryption key $dk_k$. 

In Case 2, the adversary can freely query the decryption and encryption keys as long as he respect the conditions of the no-write game: 1) He can query $dk_k$ iff $ek_k^\oneACEpow \notin ek_i$ for all $i \in I_S$, 2) He can query any encryption key after the challenge.
If the adversary $A$ can distinguish between the two hybrids, we can create an adversary $B$ that wins the alternative no-write game (Def.~\ref{def:alternativeACEnowrite}) for the \oACE scheme. 
Intuitively, this is done by letting $B$ create the ciphertext $c'$ as in Hybrid $k$, forward the challenge together with message $r$, and replacing the $k$'th \oACE ciphertext by the one he receives from the challenger.
Notice that if adversary $A$ queries the decryption key $dk_k$, then he cannot query any encryption key $ek_i$ for $i \in I_S$, which contains $ek_k^\oneACEpow$. 
Thus, adversary $B$ can query to get the decryption key $dk_k^\oneACEpow$ but does not query for the encryption key $ek_k^\oneACEpow$ until after the challenge.
This means that the \oACE ciphertext ${c'_k}^\oneACEpow$ send by the challenger is either a random \oACE ciphertext or a sanitization of $c_k^\oneACEpow$.
If $A$ wins the game with non-negligible probability, then so does the adversary $B$. This gives us a contradiction with the assumption that the \oACE scheme satisfies the no-write rule. Thus, we can conclude that $A$ cannot distinguish between the two hybrids.

If the adversary does not query the decryption key $dk_k$, then he can either query an encryption key $ek_i$, where $ek_k^\oneACEpow \in ek_i$ before the challenge (i.e.~Case 1), or after the challenge (equivalent to Case 2). 
Furthermore, notice that in both cases all encryption queries $(i,m)$ are answered by sending the message to the challenger to get the $k$'th \oACE sanitized ciphertext for the response. 

\paragraph{Hybrid $n$ $\approx$ Game 1.} 
In Hybrid $n$ the ciphertext $c'$ is generated by running $\CCG_n$. From the description of this algorithm we can conclude that this is actually a ciphertext generated by $\ACE.\San(rk,c)$. Thus, Game 1 and Hybrid $n$ both generates the ciphertext $c'$ as a sanitized version of the ciphertext $c$. Thus, Game 1 and Hybrid $n$ are identical and therefore indistinguishable to any adversary.

\end{proof}



\section{Sanitizable Functional Encryption Scheme - Proofs}

\subsection{Proof of Lemma~\ref{lem:IND-CPA_sFE}}\label{app:IND-CPA_sFE}
\setcounter{claimcounter}{0}

In this appendix we show the full adaptive proof of the IND-CPA security for Construction~\ref{con:sanFE}. The proof follows closely the selective IND-CPA security proof of the functional encryption scheme presented by Garg et. al.~\cite{DBLP:conf/focs/GargGH0SW13}. 
The main differences between the proof are the following: we enhance their proof by making it adaptive and quantifying the advantage of the adversary. Furthermore, we make a small change in the proof of the valid iO instance, which make the proof work for our sanitizable version of the functional encryption scheme.

\paragraph{Game 0.} The IND-CPA security game where $b=0$;

\paragraph{Game 1.} The IND-CPA security game where $b=1$; \\

Before proving that Game 0 and Game 1 are indistinguishable, we prove that the following sequence of hybrids are indistinguishable. 

\paragraph{Hybrid 0.} The challenger choose uniformly random two messages $m'_0,m'_1 \in \M$. Then he proceeds as in Game 0 with the exception that when he receives $m_0$ and $m_1$ from the adversary, then he checks that $m_0 = m'_0$ and $m_1 = m'_1$. If this is not the case, then the challenger aborts the game, otherwise he continues as in Game 0.

\paragraph{Hybrid 1.} The same as Hybrid 0, except that after choosing the messages $m'_0$ and $m'_1$, the challenger encrypts message $m'_0$ twice: $c_i^* = \sanPKE.\Enc(pk_i,m'_0;r_i)$ for $i=1,2$. Then the challenger simulates the common reference string $crs_E$ and NIZK proof $\pi_E^*$ as follows
	\begin{align*}
	(crs_E,\tau) \from \NIZK.\Sim_1(1^\kappa,x_E), \quad \pi_E^* \from \NIZK.\Sim_2(crs_E, \tau, x_E)
	\end{align*}
	where $x_E = (c_1^*,c_2^*)$ is the statement we want to prove (see definition of $R_E$ in Construction~\ref{con:sanFE}). Note: the challenger still checks that $m_0 = m'_0$ and $m_1 = m'_1$.

\paragraph{Hybrid 2.} The same as Hybrid 1, except that we change the message of the second PKE ciphertext. Thus, the ciphertexts are computed as follows
	\begin{align*}
	c_1^* = \sanPKE.\Enc(pk_1,m'_0;r_1), \quad c_2^* = \sanPKE.\Enc(pk_2,m'_1;r_2)
	\end{align*}
	Note: $crs_E$ and $\pi_E^*$ are still simulated.

\paragraph{Hybrid 3,$i$ for $i =0,\dots,q$.} The same as Hybrid 2, except that in the first $j \leq i$ secret key queries the secret key is generated as an obfuscation of program $P_2$, while in the last $j > i$ secret key queries the secret key is generated as an obfuscation of program $P_1$.
Observe, Hybrid 2 and Hybrid 3,0 are identical.

\begin{center}
\begin{small}
    \begin{tabular}{| p{6cm} | p{6cm} |}
	\hline
	\multicolumn{1}{|c|}{\textbf{\emph{Program} $P_1$}} & \multicolumn{1}{|c|}{\textbf{\emph{Program} $P_2$}}\\
	\hline 
	\
	
	Input: $c_1,c_2,\pi_S$; 
	
	Const: $crs_S,f,sk_1$;
	
	\

	1. If $\NIZK.\Verify(crs_S,(c_1,c_2),\pi_S) = 1$;
	\begin{itemize}
		\item[] output $f(\sanPKE.\Dec(sk_1,c_1))$;
	\end{itemize}

	2. else output fail;

	&
	\

	Input: $c_1,c_2,\pi_S$; 
	
	Const: $crs_S,f,sk_2$;
	
	\

	1. If $\NIZK.\Verify(crs_S,(c_1,c_2),\pi_S) = 1$;
	\begin{itemize}
		\item[] output $f(\sanPKE.\Dec(sk_2,c_1))$;
	\end{itemize}

	2. else output fail;

	\\
	\hline
	\end{tabular}
\end{small}
\end{center}

\paragraph{Hybrid 4.} The same as Hybrid $3,q$, except that we change the message of the first PKE ciphertext. Thus, the ciphertexts are computed as follows
	\begin{align*}
	c_i^* = \sanPKE.\Enc(pk_i,m'_1;r_i) \quad \text{for } i=1,2
	\end{align*}

\paragraph{Hybrid 5,$i$ for $i =0,\dots,q$.} The same as Hybrid 4, except that in the first $j \leq i$ secret key queries the secret key is generated as an obfuscation of program $P_1$, while in the last $j > i$ secret key queries the secret key is generated as an obfuscation of program $P_2$.
Observe, Hybrid 4 and Hybrid 5,0 is identical.

\paragraph{Hybrid 6.} The same as Hybrid $5,q$, except that $crs_E$ and $\pi_E^*$ is generated honestly, and the PKE ciphertexts are generated as in  Hybrid $5,q$, but after we see the challenge, and have checked that $m_0 = m'_0$ and $m_1 = m'_1$. \\

We  now show that each sequential pair of the hybrids are indistinguishable.

\begin{claim} \label{claim:sFE_IND-CPA_ZK_1}
For any adversary $A$ that can distinguish Hybrid 0 and Hybrid 1, there exists an adversary $B$ for the computational zero-knowledge property of the SSS-NIZK scheme such that the advantage of $A$ is
	\begin{align*}
	\adv{A} \leq \adv{\NIZK,B}
	\end{align*}
\end{claim}


\begin{proof}
Assume that any adversary breaks the computational zero-knowledge property with advantage at most $\epsilon$.
Assume for contradiction that there exists an adversary $A$ that distinguishes the hybrids with advantage greater than $\epsilon$, then we can construct a poly-time adversary $B$ that breaks the computational zero-knowledge property with advantage greater than $\epsilon$. 

$B$ begins by choosing $m'_0,m'_1 \in \M$ uniformly random. Then he computes the PKE public and secret keys honestly, and runs one copy of the NIZK setup algorithm: $crs_S \from \NIZK.\Setup(1^\kappa)$. Next, $B$ encrypts message $m'_0$ under both public keys: $c_i^* \from \sanPKE.\Enc(pk_i,m'_0;r_i)$ for $i=1,2$, and sends $x_E = (c_1^*,c_2^*)$ and $w_E = (m'_0,r_1,r_2)$ to the zero-knowledge challenger. In return $B$ receives $(crs',\pi')$, and he uses $crs_E = crs'$ and $\pi_E^* = \pi'$. Thus, $B$ lets the public parameters $pp = (crs_E,crs_S,pk_1,pk_2)$, and the challenge response $c^* = (c_1^*,c_2^*,\pi_E^*)$. 
The rest of the game follow the structure of the sFE IND-CPA game, and concludes with $B$ forwarding $A$'s response $b'$ to the challenger. 

Observe, if the challenger generates $crs'$ and $\pi'$ honestly, then we are in Hybrid 0, and if the challenger simulates the proof, then we are in Hybrid 1. Thus, if adversary $A$ can distinguish the hybrids with advantage greater than $\epsilon$, then adversary $B$ can break the computational zero-knowledge property of the SSS-NIZK scheme with advantage greater than $\epsilon$. Thus, we reach a contradiction.
\end{proof}

\begin{claim} \label{claim:sFE_IND-CPA_PKE_1}
For any adversary $A$ that can distinguish Hybrid 1 and Hybrid 2, there exists an adversary $C$ for the IND-CPA security game of the sanitizable PKE scheme such that the advantage of $A$ is
	\begin{align*}
	\adv{A} \leq \adv{\sanPKE,C}
	\end{align*}
\end{claim}

\begin{proof}
Assume that any adversary wins the IND-CPA security game with advantage at most $\epsilon$.
Assume for contradiction that there exists an adversary $A$ that distinguishes the hybrids with advantage greater than $\epsilon$, then we can construct a poly-time adversary $C$ that breaks the PKE IND-CPA security with advantage greater than $\epsilon$. 

$C$ begins by choosing $m'_0,m'_1 \in \M$ uniformly random. Then he computes $(pk_1,sk_1) \from \sanPKE.\Setup(1^\kappa)$ and $c_1^* \from \sanPKE.\Enc(pk_1,m'_0;r_1)$, and he runs one copy of the NIZK setup algorithm: $crs_S \from \NIZK.\Setup(1^\kappa)$. He then receives a second public key $pk'$ from the challenger, and sets $pk_2 = pk'$. Next, he sends $m'_0,m'_1$ to the challenger at receives back $c'$, and sets $c_2^* = c'$. Then $C$ uses the SSS-NIZK simulator to generate $crs_E$ and $\pi_E^*$ for the statement $x_E = (c_1^*,c_2^*)$. Then, $C$ sets the public parameters $pp = (crs_E,crs_S,pk_1,pk_2)$, and the challenge response $c^* = (c_1^*,c_2^*,\pi_E^*)$. 
The rest of the game follow the structure of the sFE IND-CPA game, and concludes with $C$ forwarding $A$'s response $b'$ to the challenger. 

Observe, if the challenger encrypts $m'_0$ then we are in Hybrid 1, and if the challenger encrypts $m'_1$ then we are in Hybrid 2. Thus, if adversary $A$ can distinguish the hybrids with advantage greater than $\epsilon$, then adversary $C$ can break the IND-CPA security of the PKE scheme with advantage greater than $\epsilon$. Thus, we reach a contradiction.
\end{proof}

\begin{claim} \label{claim:sFE_IND-CPA_iO_1}
For any adversary $A$ that can distinguish Hybrid $3,i$ and Hybrid $3,i+1$, there exists an adversary $D$ for iO such that the advantage of $A$ is
	\begin{align*}
	\adv{A} \leq \adv{iO,C} (1 - 2\sssprob)
	\end{align*}
	where $\sssprob$ is the negligible probability of the statistical simulation-soundness for the SSS-NIZK scheme.
\end{claim}

\begin{proof}
Assume that any adversary wins the iO indistinguishable game with advantage at most $\epsilon$.
Assume for contradiction that there exists an adversary $A$ that distinguishes the hybrids with advantage greater than $\epsilon$, then we can construct a poly-time adversary $D$ that breaks the iO property with advantage greater than $\epsilon \cdot (1 - 2\sssprob)$. 

$D$ begins by choosing $m'_0,m'_1 \in \M$ uniformly random. Then he runs two copies of the PKE setup algorithm and encrypts: $c_1^* \from \sanPKE.\Enc(pk_1,m'_0;r_1)$ and $c_2^* \from \sanPKE.\Enc(pk_2,m'_1;r_2)$. Next, he runs the SSS-NIZK setup algorithm to get $crs_S$, and uses the SSS-NIZK simulators to get $crs_E$ and $\pi_E^*$ for statement $x_E=(c_1^*,c_2^*)$. Then, $D$ sets the public parameters $pp = (crs_E,crs_S,pk_1,pk_2)$, and the challenge response $c^* = (c_1^*,c_2^*,\pi_E^*)$. 
Next, the game follows the structure of the sFE IND-CPA game, except the secret key queries are generated as follows
\begin{itemize}
\item For $j \leq i$ the $j$'th secret key is generated as the obfuscation of program $P_2$
\item For $j > i+1$ the $j$'th secret key is generated as the obfuscation of program $P_1$
\item For $j=i+1$ the $j$'th secret key is generated as follows: $D$ sends $P_1$ and $P_2$ to the challenger, and receives back an obfuscated program $P'$, which $D$ sends to $A$ as the $i+1$'th secret key.
\end{itemize}
$D$ concludes the game by forwarding $A$'s response $b'$ to the challenger.  \\

Assume that the two programs $P_1$ and $P_2$ are a valid iO instance, then we observe: if the challenger obfuscates program $P_1$, then we are in Hybrid $3,i$, and if the challenger obfuscates program $P_2$, then we are in Hybrid $3,i+1$. Thus, if adversary $A$ can distinguish the hybrids with advantage greater than $\epsilon$, then adversary $D$ can break the iO property with advantage greater than $\epsilon$. 

If the two program are not a valid iO instance, then $D$ does not follow the rules of the iO game, and cannot win the game. Let $p$ be the probability that the two programs are a valid iO instance. Then the adversary $D$ can break the iO property with advantage greater than $\epsilon p$. \\

Next, we prove that $P_1$ and $P_2$ are a valid instance with probability $p \geq  1-2\sssprob$.
To prove this we look at the possible inputs to the program. Let $c'=(c'_1,c'_2,\pi_S)$ denote the sanitization of $c=(c_1,c_2,\pi_E)$.
\begin{enumerate}
\item $c$ is a correct encryption of some message (i.e., $c_1$ and $c_2$ are encryption of the same message, and $\pi_E$ is a proof of that), and $c'$ is a correct sanitization of ciphertext $c$ (i.e., $c'_i$ is a PKE sanitization of $c_i$ for $i=1,2$, and $\pi_S$ is a proof of that and a proof that $\pi_E$ is a correct proof).
\item $c_1$ and $c_2$ are PKE encryptions of different messages, but the proof $\pi_E$ verifies, and $c'$ is a correct sanitization of ciphertext $c$.
\item $c$ is an incorrect ciphertext such that the proof $\pi_E$ does not verify, and $c'$ is a correct sanitization of $c$.
\item $c'$ is an incorrect sanitized ciphertext such that the proof $\pi_S$ does not verify.
\end{enumerate}

Case 1: the proof $\pi_S$ passes the verification, and both programs  decrypt to the same message and compute the same function. 

Case 2: the statistical simulation-soundness of the SSS-NIZK used in the encryption gives us that with probability negligible close to one, this only happens if $c_1 = c_1^*$ and $c_2 = c_2^*$ (i.e. the ciphertext for which the proof was simulated). Since the simulated proof can be verified, we can construct a correct proof $\pi_S$ that also verifies. Thus, program $P_1$ decrypts $c_1$ to $m'_0$ and outputs $f_{i+1}(m'_0)$, while $P_2$ decrypts $c_2$ to $m'_1$ and outputs $f_{i+1}(m'_1)$. Thus, the two program have the same output, since $f_{i+1}(m'_0) = f_{i+1}(m'_1)$. 
This means that the two programs have different output on the same input with negligible probability $\sssprob$.

Case 3: the statistical soundness property of the SSS-NIZK used in the sanitization gives us that this is impossible with probability negligible close to one. Since the proof $\pi_E$ does not verify, we cannot construct a proof $\pi_S$ that verifies. 
Thus, the two programs have different output on the same input with negligible probability $\sssprob$.

Case 4: The verification of the proof $\pi_S$ does not pass in both programs. Thus, both programs outputs fail. \\

This means, that the combined probability that the two programs have different output on the same input can be upper bounded by  $2\sssprob$. Thus, the probability that the two programs are a valid iO instance is $p \geq 1-2\sssprob$.
\end{proof}

\begin{claim} \label{claim:sFE_IND-CPA_PKE_2}
For any adversary $A$ that can distinguish Hybrid $3,q$ and Hybrid 4, there exists an adversary $C$ for the IND-CPA security game of the sanitizable PKE scheme such that the advantage of $A$ is
	\begin{align*}
	\adv{A} \leq \adv{\sanPKE,C}
	\end{align*}
\end{claim}

The proof of the claim follow the same structure at the proof of Claim~\ref{claim:sFE_IND-CPA_PKE_1}.

\begin{claim} \label{claim:sFE_IND-CPA_iO_2}
For any adversary $A$ that can distinguish Hybrid $5,i$ and Hybrid $5,i+1$, there exists an adversary $D$ for iO such that the advantage of $A$ is
	\begin{align*}
	\adv{A} \leq \adv{iO,C} (1 - 2\sssprob)
	\end{align*}
	where $\sssprob$ is the negligible probability of the statistical simulation-soundness for the SSS-NIZK scheme. \\
\end{claim}

The proof of the claim follow the same structure at the proof of Claim~\ref{claim:sFE_IND-CPA_iO_1}.

\begin{claim} \label{claim:sFE_IND-CPA_ZK_2}
For any adversary $A$ that can distinguish Hybrid $5,q$ and Hybrid 6, there exists an adversary $B$ for the computational zero-knowledge property of the SSS-NIZK scheme such that the advantage of $A$ is
	\begin{align*}
	\adv{A} \leq \adv{\NIZK,B}
	\end{align*}
\end{claim}

The proof of the claim follow the same structure at the proof of Claim~\ref{claim:sFE_IND-CPA_ZK_1}. \\

From these claims we can conclude that for any adversary $A$ that can distinguish Hybrid 0 and Hybrid 6, there exists an adversary $B$ for the computational zero-knowledge property of the SSS-NIZK scheme, an adversary $C$ for the IND-CPA security game of the sanitizable PKE scheme, and an adversary $D$ for iO such that the advantage of $A$ is
	\begin{align*}
	\adv{A} \leq 2\cdot\adv{\NIZK,B} + 2\cdot\adv{\sanPKE,C} + 2q\cdot\adv{iO,C} (1 - 2\sssprob)
	\end{align*}
	where $q$ is the number of the secret key queries the adversary makes during the game.

We conclude the proof of Lemma~\ref{lem:IND-CPA_sFE} by proving that for any adversary $E$ that can distinguish Game 0 and Game 1, there exists an adversary $A$ that can distinguish Hybrid 0 and Hybrid 6 such that the advantage of $E$ is $$\adv{\sanFE,E} \leq 2|\M| \cdot\adv{A}$$

Assume that any adversary $E$ can distinguish Game 0 and Game 1 with advantage $\epsilon$. 
We start by changing the game such that the challenger guesses the two messages in the beginning of the games, and aborts if the guesses is wrong. Thus, the advantage of the adversary is now $\epsilon/2|\M|$.
Observe, the two new games are identical to Hybrid 0 and Hybrid 6. Thus, we get the inequality which concludes the proof.
%

\subsection{Proof of Lemma~\ref{lem:Sanitize_sFE}}\label{app:Sanitize_sFE}
\setcounter{claimcounter}{0}

In this appendix we show that Construction~\ref{con:sanFE} fulfils the sanitizable property from Definition~\ref{def:FEnowrite}.
Lets define the following games
\paragraph{Game 0.} The sanitization game where $b=0$;
\paragraph{Game 1.} The sanitization game where $b=1$; \\

Before proving that Game 0 and Game 1 are indistinguishable, we  prove that the following sequence of hybrids are indistinguishable. 

\paragraph{Hybrid 0.} The challenger chooses a uniformly random message $m \in \M$. Then he proceeds as in Game 0 with the exception that when he receives $c$ from the adversary, he first checks that $\MasterDec(msk,c) = m$. If this is not the case, then the challenger aborts the game, otherwise he continues as in Game 0.

\paragraph{Hybrid 1.} The same as Hybrid 0, except that the challenger encrypts the message $m$ twice and sanitizes the two encryptions
	\begin{align*}
	c_i &\from \sanPKE.\Enc(pk_i,m;r_i) \quad \text{for } i = 1,2 \\
	c^*_i &\from \sanPKE.\San(pk_i,c_i;s_i) \quad \text{for } i = 1,2 
	\end{align*}
	When the challenger receives the challenge $c$ he checks that $\MasterDec(msk,c) = m$, constructs the proof $\pi_S$ honestly, and sends the respond: $(c^*_1,c^*_2,\pi_S^*)$.

\paragraph{Hybrid 2.} The same as Hybrid 1, except that after choosing the messages $m$, the challenger generates the PKE keys $(pk_i,sk_i) \from \sanPKE.\Setup(1^\kappa)$ and generates $crs_E \from \NIZK.\Setup(1^\kappa)$. 
Next, the challenger encrypts the message $m$ twice and sanitizes the two encryptions
	\begin{align*}
	c_i &\from \sanPKE.\Enc(pk_i,m;r_i) \quad \text{for } i = 1,2 \\
	c^*_i &\from \sanPKE.\San(pk_i,c_i;s_i) \quad \text{for } i = 1,2 
	\end{align*}
Then the challenger simulates the common reference string $crs_S$ and NIZK proof $\pi_S^*$ as follows
	\begin{align*}
	(crs_S,\tau) \from \NIZK.\Sim_1(1^\kappa,x_S), \quad \pi_S^* \from \NIZK.\Sim_2(crs_S, \tau, x_S)
	\end{align*}
	where $x_S = (c_1^*,c_2^*)$ is the statement we want to prove (see definition of $R_S$ in Construction~\ref{con:sanFE}). 
Thus, the public parameters are $pp = (crs_E, crs_S, pk_1, pk_2)$, the master secret key is $msk = sk_1$, and the challenger response is $(c^*_1,c^*_2,\pi_S^*)$.
When receiving the challenge $c$ from the adversary, the challenger still checks that $\MasterDec(msk,c) = m$.

\paragraph{Hybrid 3.} The same as Hybrid 2, except that after choosing the messages $m$ we generate the system parameters honestly and send them to the adversary. 
Next and before seeing the challenge, the challenger generates two PKE encryptions of the message $m$ and sanitizes them as in Hybrid 1 to get $c^*_1$ and $c^*_2$. Then he generates the proofs $\pi_E$ and $\pi_S^*$ honestly. Thus, the challenge response is $c^* = (c^*_1,c^*_2,\pi_S^*)$.
After receiving the challenge $c$, the challenger still checks that $\MasterDec(msk,c) = m$ before sending $c^*$.

\paragraph{Hybrid 4.} The challenger chooses a uniformly random message $m \in \M$. Then he proceeds as in Game 1 with the exception that when he receives $c$ from the adversary, then he checks that $\MasterDec(msk,c) = m$. \\

We now show that each sequential pair of the hybrids are indistinguishable.

\begin{claim} \label{claim:sFE_san_PKE}
Hybrid 0 and Hybrid 1 are identical.
\end{claim}

\begin{proof}
Let $e_1$ and $e_2$ be the PKE encryptions from the challenge $c$. The perfect sanitization property of the PKE scheme (see Definition~\ref{def:PKE_sanitation}) states that given $c^*_1$, $c^*_2$, $e_1$, and $e_2$ there exists some randomness $s'_1$ and $s'_2$ such that
	\begin{align*}
	c^*_i = \sanPKE.\San(pk_i,e_i;s'_i) \quad \text{for } i=1,2
	\end{align*}
Thus, we can conclude that two hybrids are identical (note that this randomness only need to exist for the argument to go through, not to be efficiently computable). 
\end{proof}

\begin{claim} \label{claim:sFE_san_ZK_1}
For any adversary $A$ that can distinguish Hybrid 1 and Hybrid 2, there exists an adversary $B$ for the computational zero-knowledge property of the SSS-NIZK scheme such that the advantage of $A$ is 
\begin{align*}
	\adv{\sanFE,A} \leq \adv{\NIZK,B} 
\end{align*}
\end{claim}

\begin{proof} 
Assume that any adversary breaks the computational zero-knowledge property with advantage at most $\epsilon$. Assume for contradiction that there exists an adversary $A$ that distinguishes the hybrids with advantage greater than $\epsilon$, then we can construct an adversary $B$ that breaks the computational zero-knowledge property with advantage greater than $\epsilon$.

$B$ begins by choosing $m \in \M$ uniformly at random. Then he computes the PKE keys $pk_1,pk_2$ and $sk_1$ honestly and generates $crs_E \from \NIZK.\Setup(1^\kappa)$. Next, he generates an honest sFE encryption of $m$: $(c_1,c_2,\pi_E) \from \Enc(pp,m)$, and sanitizes the PKE ciphertexts: $c^*_i \from \sanPKE.\San(pk_i,c_i;s_i)$ for $i=1,2$. 
Then it sends the statement $x_S = (c^*_1,c^*_2)$ and the witness $w_S = (c_1,c_2,s_1,s_2,\pi_E)$ to the challenger, and receives back a common reference string $crs'$ and a proof $\pi'$. 
$B$ then sets $crs_S = crs'$ and $\pi_S^* = \pi'$. Thus, the public parameters are $pp = (crs_S,crs_E,pk_1,pk_2)$, the master secret key is $msk=sk_1$, and the challenger response is $c^* = (c^*_1,c^*_2,\pi_S^*)$. 
$B$ sends $pp$ and $msk$ to adversary $A$, receives a challenge $c$, checks that $\MasterDec(msk,c) = m$, and if so respond with $c^*$. 

If the challenger generates $crs'$ and $\pi'$ honestly, then we are in Hybrid 1, and if the challenger simulates the proof, then we are in Hybrid 2. 
Thus, if adversary $A$ can distinguish between the hybrids with advantage greater than $\epsilon$, then adversary $B$ can break the computational zero-knowledge property with advantage greater than $\epsilon$.
\end{proof}

\begin{claim} \label{claim:sFE_san_ZK_2}
For any adversary $A$ that can distinguish Hybrid 2 and Hybrid 3, there exists an adversary $B$ for the computational zero-knowledge property of the SSS-NIZK scheme such that the advantage of $A$ is 
\begin{align*}
	\adv{\sanFE,A} \leq \adv{\NIZK,B} 
\end{align*}
\end{claim}

The proof of this claim follows the same structure as the proof of Claim~\ref{claim:sFE_san_ZK_1}. 

\begin{claim} \label{claim:sFE_san_identical}
Hybrid 3 and Hybrid 4 are identical. 
\end{claim}

\begin{proof} 
In both hybrids, we check that the received challenge $c$ is an encryption of the message $m$, we guessed in the beginning of the game. 
In Hybrid 3, we create an honest encryption and sanitization of the message $m$ before seeing the challenge $c$. 
In Hybrid 4, we decrypt the challenge $c$ to get the message $m$ (same as the one we guessed), then we create an honest encryption and sanitization of the message. 
Thus, the two hybrids are identical. 
\end{proof}

From these claims we can conclude that for any adversary $A$ that can distinguish Hybrid 0 and Hybrid 4, there exists an adversary $B$ for the computational zero-knowledge property of the SSS-NIZK scheme such that the advantage of $A$ is at most $2 \cdot\adv{\NIZK,B}$.

We conclude the proof by proving that for any adversary $A$ that can distinguish Game 0 and Game 1, there exists an adversary $B$ that can distinguish Hybrid 0 and Hybrid 4 such that the advantage of $A$ is
	\begin{align*}
	\adv{\sanFE,A} \leq 2|\M| \cdot\adv{\NIZK,B} 
	\end{align*}

Assume that any adversary $A$ can distinguish Game 0 and Game 1 with advantage $\epsilon$. 
We start by changing the game such that the challenger guesses the message in the beginning of the games, and aborts if the guesses is wrong. Thus, the advantage of the adversary is now $\epsilon/|\M|$.
Observe, the two new games are identical to Hybrid 0 and Hybrid 4. Thus, we get the above inequality.

\end{document}